\documentclass[11p,reqno]{amsart}

\usepackage[T1]{fontenc}
\usepackage{graphicx}
\usepackage{amssymb,amsthm,amsmath,mathrsfs,bm,braket,marginnote}
\usepackage{enumerate}
\usepackage{appendix}
\usepackage[colorlinks=true, pdfstartview=FitV, linkcolor=blue, citecolor=blue, urlcolor=blue]{hyperref}
\usepackage{multirow}

\usepackage{pgf}
\usepackage{pgfplots}
\usepackage{tikz}
\usetikzlibrary{arrows,calc}
\usepackage{verbatim}
\usetikzlibrary{decorations.pathreplacing,decorations.pathmorphing}
\usepackage[numbers,sort&compress]{natbib}
\usepackage{dsfont}

\numberwithin{equation}{section}
\newtheorem{theorem}{Theorem}[section]

\newtheorem{coro}[theorem]{Corollary}

\newtheorem{proposition}[theorem]{Proposition}

\theoremstyle{remark}
\newtheorem{remark}[theorem]{Remark}

 \reversemarginpar

\newcommand{\lp}{e^{s x}}
\newcommand{\ma}{\mathcal{A}}
\newcommand{\mb}{\mathcal{B}}

\newcommand{\mt}{\mathcal{T}}
\newcommand{\mf}{\mathcal{F}}

  \newcommand{\floor}[1]{\lfloor #1 \rfloor}
  
\begin{document}

\title[$q$-Pearson pair and $q$-deformed ensemble]{$q$-Pearson pair and moments in $q$-deformed ensembles}

\date{}

\author{Peter J. Forrester}
\address{School of Mathematical and Statistics, ARC Centre of Excellence for Mathematical and Statistical Frontiers, The University of Melbourne, Victoria 3010, Australia}
\email{pjforr@unimelb.edu.au}

\author{Shi-Hao Li}
\address{Department of Mathematics, Sichuan University, Chengdu, 610064, China}
\email{lishihao@lsec.cc.ac.cn}

\author{Bo-Jian Shen}
\address{School of Mathematical Sciences, Shanghai Jiaotong University, People's Republic of China.}
\email{JOHN-EINSTEIN@sjtu.edu.cn}

\author{Guo-Fu Yu}
\address{School of Mathematical Sciences, Shanghai Jiaotong University, People's Republic of China.}
\email{gfyu@sjtu.edu.cn}

\dedicatory{}


\begin{abstract}
The generalisation of continuous orthogonal polynomial ensembles from random matrix theory to the $q$-lattice setting is considered. We take up the task of initiating a systematic study of the corresponding moments of the density from two complementary viewpoints. The first requires knowledge of the ensemble average with respect to a general Schur polynomial, from which the spectral moments follow as a corollary. In the case of little $q$-Laguerre weight, a particular ${}_3 \phi_2$ basic hypergeometric polynomial is used to express density moments. The second approach is to study the $q$-Laplace transform of the un-normalised measure. Using integrability properties associated with the $q$-Pearson equation for the $q$-classical weights, a fourth order $q$-difference equation is obtained, generalising a result of Ledoux in the continuous classical cases.
\end{abstract}

\maketitle

\section{Introduction}

\subsection{Special properties of moments for classical ensembles with unitary symmetry}
An ensemble of $N\times N$ complex Hermitian random matrices $\{ H \}$ is said to have
unitary symmetry if the probability density function (PDF) has the invariance
property $P( U^{-1} H U) = P(H)$ for all unitary matrices $U\in\mathcal{U}(N)$. A structurally simple example
is when, up to proportionality, $P(H)$ is given by $\prod_{l=1}^N w(x_l)$, where
$\{ x_l \}_{l=1}^N$ are the eigenvalues of $ H $ and $w(x)$ is referred to as the weight function.
In this circumstance, up to normalisation, the corresponding eigenvalue PDF is given
by (see \cite[Prop.~1.3.4]{forrester10})
\begin{equation}\label{1.1}
\prod_{l=1}^N w(x_l) \prod_{1 \le j < k \le N} | x_k - x_j |^2.
\end{equation}

The class of PDFs of the form (\ref{1.1}) are closely related to orthogonal polynomials. Thus let
$\{ p_k(x) \}_{k=0,1,\dots}$, where each $p_k(x)$ is a monic polynomial of degree
$k$, have the orthogonality
\begin{equation}\label{1.2}
\int_I w(x) p_j(x) p_k(x) \, dx = h_j \delta_{j,k},
\end{equation}
where $I$ denotes the support of $w(x)$, and $h_j > 0$ is the normalisation.
It is a standard result in random matrix theory (see \cite[Prop.~5.1.2]{forrester10})
that 
\begin{equation}\label{1.2a}
\rho_{(k), N}(x_1,\dots, x_k) = \det  [ K_N(x_j, x_l) ]_{j,l=1}^k,
\end{equation}
where
\begin{equation}\label{1.2b}
K_N(x, y) = \Big ( w(x) w(y) \Big )^{1/2} \sum_{j=0}^{N-1} {p_j(x) p_j(y) \over h_j}.
\end{equation}
Here $\rho_{(k), N}$ denotes the $k$-point correlation function, obtained by integrating
a suitably normalised version of (\ref{1.1}) over all variables except $x_1,\dots, x_k$.
In the case $k=1$ this corresponds to the eigenvalue density, and we read off from
(\ref{1.2a}) and (\ref{1.2b}) that
\begin{equation}\label{1.2c}
\rho_{(1), N}(x) = w(x) \sum_{j=0}^{N-1} {( p_j(x))^2 \over h_j}.
\end{equation}

A distinguished class of weight functions have the property that their logarithmic derivative 
can be written in the particular rational function form
\begin{equation}\label{1.1a}
\frac{d}{dx}\log w(x) = {a_0 + a_1 x \over b_0 + b_1 x + b_2 x^2},
\end{equation}
for some polynomials of degree less than or equal to one in the numerator
and two in the  denominator, as indicated. Requiring too that $w(x)$ has all moments
finite leaves just three possibilities, up to linear scalings,
\begin{equation}\label{1.1b}
w^{(G)}(x) = e^{- x^2}, \qquad w^{(L)}(x) =  x^a e^{-x} \chi_{x > 0}, \qquad w^{(J)}(x) = x^a (1 - x)^b
\chi_{0 < x < 1},
\end{equation}
referred to as the Gaussian, Laguerre and Jacobi weights respectively. The notation
$\chi_A$ used in (\ref{1.1b}) is the indicator function of the condition $A$, which takes
on the value one for $A$ true, and zero otherwise. For the parameters in the Laguerre and
Jacobi case we require $a,b > -1$ for the weights to be normalisable.
Note that the corresponding orthogonal
polynomials in (\ref{1.2}) are, up to normalisation,
 the classical Hermite, Laguerre and Jacobi polynomials \cite{Sz75}. The corresponding
 (classical) random matrix ensembles
 with eigenvalue PDF (\ref{1.1}) are given the names GUE, LUE and JUE respectively, with the
 first letter corresponding to the weight, and UE denoting unitary ensemble. Here unitary is
 used in the context of the symmetry noted in the first paragraph.
 
 For a non-negative integer $ k $, the moments associated with the spectral density are defined by
 \begin{equation}\label{1.1c}
 m_{k,N} = \int_I x^k \rho_{(1), N}(x) \, dx.
 \end{equation}
 There are special properties associated with the moments
 in the case of the classical ensembles.
 Consider for example the GUE.
 It is a long established result in the applications of random matrices, due to
 Harer and Zagier \cite{HZ86}, that  the moments (\ref{1.1c}) have combinatorial and topological
 significance. This come about through the large $N$ terminating expansion
  \begin{equation}\label{1.5}
  {2^k \over N^{1 + 2k}} m_{2k,N}^{(G)} = \sum_{g=0}^{[k/2]}
  {c(g;k) \over N^{2g}}.
  \end{equation}
  It was shown in  \cite{HZ86} that the coefficients $c(g;k)$ can be specified as the number of pairings of the 
  edges of a $2k$-gon giving a figure that can be embedded on a surface of genus $g$.
  The leading coefficient in (\ref{1.5}) is given by
  \begin{equation}\label{1.5a}  
  c(0;k) = {1 \over k + 1} \binom{2k}{k},
    \end{equation}
    which is the $k$-th Catalan number.
    
    An analogous expansion  to (\ref{1.5}) holds for the LUE
    with Laguerre parameter $a = \alpha N$ \cite{Di02,CDO21}.
    The latter of these references involves so-called double monotone Hurwitz
    numbers, which have both combinatorial and topological significance.
     The recent work \cite{GGR21} considers the coefficients in the $1/N^2$
    expansion of the moments for the ensemble (\ref{1.1}) in the case of the JUE with
    Jacobi parameters $a = \alpha_1 N$,
    $b = \alpha_2 N$. Unlike the Gaussian and Laguerre cases, this expansion no longer
    terminates. It is shown that the coefficients can be expressed in terms of
  triple monotone Hurwitz numbers. 
  
  Another significant feature of the moments of the GUE and LUE
   is that they permit evaluations in terms of hypergeometric polynomials.
  Thus for the GUE \cite[Eq.~(4.33)]{WF14}
  \begin{equation}\label{1.5b}   
  {2^{2k} \over N} \int_{-\infty}^\infty |x|^{2k} \rho_{(1),N}^{(G)}(x) \, dx =
  {\Gamma(2k+1) \over \Gamma(k+1)} \, {}_2 F_1 \Big ({-k, 1-N \atop 2} \Big | 2 \Big).
   \end{equation}
   This result in fact remains valid for complex moments Re$\, k > -1/2$ \cite{cunden19}.
   In the case of the LUE, it is known \cite[Eq.~4.11]{cunden19}
     \begin{equation}\label{1.5c}  
     m_{k,N}^{(L)} = N (N + a) {(k+a)! \over (1 + a)!} \,
     {}_3 F_2 \Big ({1-k, 2+k, 1-N \atop 2, 2+a} \Big | 1 \Big ),
     \end{equation}
     which also extends to continuous $k$.
     
     \subsection{$q$-classical ensembles with unitary symmetry and their moments}\label{S1.2}
     It is well known that the classical polynomials admit $q$-generalisations catalogued
     according to the Askey scheme; see e.g.~\cite{ismail05,koekoek10}.
     Distinguishing the weight functions for classical $q$-orthogonal
   polynomials is that they satisfy the $q$-Pearson equation 
  \begin{equation}\label{16b}
  D_q(\sigma(x) w(x))=\tau(x) w(x), \qquad  D_q  f(x) = {f(x) - f(q x) \over (1 - q) x},
  \end{equation}  
  where it is required that $\sigma(x)$ be a polynomial
   of degree less than or equal to 2, and $\tau(x)$ be a polynomial of degree less than or
   equal to 1. Note that  with $D_q$ replaced by a derivative with respect
   to $x$, (\ref{16b}) is equivalent to (\ref{1.1a}). Our interest in the present article is
   in the moments of the density corresponding to the probability density function
   (\ref{1.1}) with $w(x)$ being given by a classical $q$-weight, and the support of the measure
   being appropriately chosen.
   
   It would seem
   that  the first  study of this type \cite{On81,AO84} was for the particular deformation of (\ref{1.1})
   from the real line to the unit circle with eigenvalue PDF proportional to
 \begin{equation}\label{4.a}   
 \prod_{l=1}^N   \theta_3(e^{i \theta_l} ;q)   \prod_{1 \le j < k \le N} | e^{i \theta_k} -  e^{i \theta_j}  |^2.
\end{equation}
Here
 \begin{equation}\label{t3}
  \theta_3(z;q) = \sum_{n=-\infty}^\infty q^{n^2} z^n, \qquad |q| < 1
  \end{equation}
  is a Jacobi theta function. This came about through the study of a certain solvable $U(N)$ lattice
  gauge theory in two-dimensions.
  
  Let the eigenvalue density associated with (\ref{4.a}) be denoted $ \rho_{(1),N}^{(RS)}(\theta)$.
  Here the superscript $(RS)$ indicates that (\ref{4.a}) relates to the Rogers-Szeg\"o polynomials
  in the circular analogue of (\ref{1.2}).
  Define the moments
   \begin{equation}\label{t3a}
   m_{k,N}^{(RS)}(q) = \int_0^{2 \pi} e^{i k \theta}  \rho_{(1),N}^{(RS)}(\theta) \, d \theta, \qquad k \in \mathbb Z.
 \end{equation}
 Then we have from    \cite{On81,AO84} that
  \begin{equation}\label{t3b} 
     m_{k,N}^{(RS)}(q) = -  {( - q)^k \over 1 - q^{2k}} \,
  {}_2 \phi_1  \Big ( {q^{2k}, q^{-2k} \atop q^{2}} \Big | q^{2}; q^{2N + 2} \Big ).
  \end{equation}
 Here the special function on the RHS refers to Heine's $q$-generalisation of the
 Gauss hypergeometric function,
  \begin{equation}\label{14a+}
  {}_2 \phi_1 \Big ( {a_1, a_2 \atop b_1} \Big | q;z \Big ) := \sum_{n=0}^\infty {(a_1;q)_n (a_2; q)_n \over (q;q)_n (b_1;q)_n} z^n,  
  \end{equation}
  with
   \begin{equation}\label{14b+}   
   (u;q)_n := (1 - u) (1 - q u ) \cdots (1 - q^{n-1} u).
  \end{equation}   
  Moreover it is noted in \cite{On81,AO84} that with the scaling $q = e^{-\lambda/N}$, the moments
  permit a $1/N^2$ expansion (recall (\ref{1.5}))
  \begin{equation}\label{14c+}   
 {1 \over N}  m_{k,N}^{(RS)}(q) \Big |_{q = e^{-\lambda/N}} =  \mu_{k,0}^{(RS)}(\lambda) + {1 \over N^2}  \mu_{k,2}^{(RS)}(\lambda) + {1 \over N^4}
  \mu_{k,4}^{(RS)}(\lambda) + \cdots
   \end{equation} 
   with
  \begin{equation}\label{14d+}   
   \mu_{k,0}^{(RS)}(\lambda)  = - {(-1)^k \over 2 \lambda k } \,   {}_2 F_1 \Big ({-k, k \atop 1} \Big | e^{-2\lambda} \Big).
 \end{equation} 
 ~\\
  Attracting attention in the recent literature \cite{forrester202,CZ21} is the particular example
 of (\ref{1.1}) proportional to
  \begin{equation}\label{15a}
  \prod_{l=1}^N w^{(SW)}(u_l;q) \prod_{1 \le j < k \le N} (u_k - u_j)^2, \qquad u_l \in \mathbb R^+,
  \end{equation}
  where
   \begin{equation}\label{15b}     
 w^{(SW)}(u;q) = {k \over \sqrt{\pi}} e^{- k^2 (\log u)^2}, \qquad q = e^{-1/(2k^2)},
   \end{equation}
   which is (one form of) the Stieltjes-Wigert weight from the theory of orthogonal
   polynomials \cite{Sz75}. Note that for the weight (\ref{15b}),
   \begin{equation}\label{15b+}  
   \int_0^\infty u^n  w^{(SW)}(u;q) \, du = q^{-(n+1)^2/2}.
    \end{equation} 
   The ensemble (\ref{15a}) turns out to be closely related to
   (\ref{4.a}). In particular, it was shown in \cite[Prop.~1.1]{forrester202}  that the moments $ m_{k,N}^{(SW)}(q)$ of the spectral
   density corresponding to (\ref{15a}) have, for $k \ge 1$, the evaluation
   \begin{equation}\label{15c}   
 {1 \over N}  q^{N k} m_{k,N}^{(SW)}(q) = 
   - {1 \over N} {( - q^{-1/2} )^k \over 1 - q^{-k}} \,
  {}_2 \phi_1  \Big ( {q^{k}, q^{-k} \atop q^{-1}} \Big | q^{-1}; q^{-N - 1} \Big )
     \end{equation} 
 (cf.~(\ref{t3b})).  Furthermore,  scaling $q = e^{-\lambda/N}$, the same $1/N^2$ expansion of the RHS results as
 in (\ref{14c+}), but with $-2 \lambda$ replaced by $\lambda$ throughout.
 
There is an equivalent formulation of the result (\ref{15c}) which replaces the continuous integral giving an average over
the probability density function by a Jackson $q$-integral. First recall that the Jackson integral with terminals 0 to $\infty$ is
defined by
  \begin{equation}\label{15d} 
  \int_0^\infty f(x) \, d_q x = (1 - q) \sum_{k=-\infty}^\infty f(q^k) q^k.
  \end{equation} 
By introducing the quantity
   \begin{equation}\label{15e} 
   c_q = (-q,-1,q;q)_\infty, \qquad (\alpha,\beta,\gamma;q)_\infty:= \prod_{l=0}^\infty (1 - \alpha q^l) (1 - \beta q^l) (1 - \gamma q^l),
  \end{equation}   
we know from    \cite{christiansen03} that with
   \begin{equation}\label{15f}  
  {w}^{( \widetilde{SW})}(x;q) = {1 \over (1 - q) \sqrt{q} c_q} x^{-1/2} e^{((\log x)/(2 \log q))^2},
  \end{equation}   
  we have
   \begin{equation}\label{15g}     
   \int_0^\infty  {w}^{( \widetilde{SW})}(x;q)  (q^{-1/2}x)^n \,  d_q x = q^{-(n+1)^2/2}.
   \end{equation}   
   This is precisely the same as for (\ref{15b+}), which is permitted since the quadratic exponential
   order of the rate of increase of the moments means they do not uniquely determine a weight function.
   As a consequence, the ensemble specified by replacing each $ w^{(SW)}(u_l;q)$ by
 $w^{( \widetilde{SW})}(u_l;q)$, and with the $u_l$ restricted to the $q$-lattice in the sense of 
 (\ref{15d}), we have the moments of this particular Jackson integral ensemble when multiplied by $q^{-k}$
 are also specified by (\ref{15c}).

 \subsection{Moments via Schur averages in the $q$ case}
 The method used to derive (\ref{15c}) in  \cite{forrester202} made use of knowledge
 of a more general result, namely the  closed form evaluation of  \cite{DT07}
  \begin{equation}\label{17a} 
  \langle s_\kappa(x_1,\dots,x_N) \rangle^{(SW)},
  \end{equation}
  where $s_\kappa$ denotes the Schur polynomial indexed by the partition $\kappa$,
  \begin{equation}\label{17b}   
   s_\kappa(x_1,\dots,x_N)  := {\det [ x_k^{\kappa_{N-j+1} + j - 1} ]_{j,k=1}^N \over \det [ x_k^{j-1} ]_{j,k=1}^N }.
  \end{equation}  
  Its relevance is seen from the identity (see e.g.~\cite{Ma95})
   \begin{equation}\label{13a+} 
 \sum_{j=1}^N x_j^k = \sum_{r=0}^{{\rm min} \, (k-1,N-1)} (-1)^r s_{(k-r,1^r)} (x_1,\dots, x_N),
 \end{equation}
 where $(k-r,1^r)$ denotes the partition with largest part $\kappa_1 = k  - r$, $r$ parts $(r \le N - 1)$ equal to
 1 and the remaining parts equal to 0, and thus
   \begin{equation}\label{13b+}  
   \int_0^\infty   x^{k} \rho_{(1),N}(x) \, d_q x =  \sum_{r=0}^{{\rm min} \, (k-1,N-1)} (-1)^r  \langle
 s_{(k-r,1^r)} (x_1,\dots, x_N) \rangle .
  \end{equation}  
  This is the same strategy as used in \cite{AO84} to derive (\ref{t3b}).
  
  In Section \ref{S2} we will discuss in detail particular $q$-weights which give evaluations of
  $  \langle s_\kappa \rangle $ simple enough that $\langle   \sum_{j=1}^N x_j^k  \rangle$
  as implied by (\ref{13b+}) is in a structured form suitable for further analysis. For example,
  as already remarked the evaluation (\ref{15c}) permits computation of the $N \to \infty$
  limit, scaled by requiring $q= e^{-\lambda/N}$. In subsection \ref{S2.1} we revisit the case of 
  Stieltjes-Wigert weight, specifically in the form (\ref{15f}), and show that the generating function
  in $N$ of the moments $m_{k,N}^{( \widetilde{SW})}$ admit the simple product form 
  (\ref{34}), a result which was conjectured recently in \cite{CZ21}. In subsection \ref{S2.2} we consider
  the discrete $q$-Hermite weight, which in a certain scaling, limits to the Gaussian weight in
  (\ref{1.1b}) for $q \to 1^-$. The corresponding probability density function (\ref{1.1}) is supported
  on the $q$-lattice corresponding to the Jackson integral
  \begin{equation}\label{16d}   
  \int_a^b f(x) \, d_qx := \int_0^b f(x) \,  d_qx - \int_0^a f(x) \,  d_qx 
 \end{equation} 
 with 
   \begin{equation}\label{16e}     
   \int_0^\alpha f(x) \, d_qx := (1 - q) \alpha \sum_{n=0}^\infty q^n f(\alpha q^n)
 \end{equation} 
 in the case $a=-1$, $b=1$.
  The moments in this case have also been the subject of some
  previous literature \cite{Wi12,PV14,MPS20}. As a new result, we use knowledge of the generating
  function in $N$ of $m_{k,N}^{(d\mbox{-}qH)}$, obtained in the recent work \cite{MPS20}, and
  rederived by our own working below, can be used to deduce the explicit form of the leading term in the
  corresponding scaled $1/N^2$ expansion; see (\ref{14c+H}). In subsection \ref{S2.3} the little
  $q$-Laguerre weight is considered. The corresponding ensemble (\ref{1.1}) is supported
  on the $q$-lattice implied by (\ref{16d}) with $a=0$, $b=1$. Upon a certain scaling, this reclaims the Laguerre weight in
(\ref{1.1b}) for $q \to 1^-$. Our main result is a closed-form evaluation of the moments in terms of
a particular ${}_3 \phi_2$ basic hypergeometric function.  This is presented in Proposition
\ref{P2.12}.

 \subsection{Integrable structures associated with $q$-orthogonal polynomial weights}

 Combining  (\ref{1.1c}) and (\ref{1.2c}) shows that the exponential generating function of the moments
 can be expressed as a sum according to
 \begin{equation}\label{18a} 
 \int_I e^{t x} \rho_{(1),N}(x) \, dx   =   \sum_{j=0}^{N-1} {1 \over h_j} \int_I e^{tx} (p_j(x))^2  \, d\mu, \quad d \mu = w(x) \, dx.
  \end{equation}
  In the works  \cite{ledoux04,ledoux05} Ledoux undertook a study of integrability properties
  associated with the integral in the summand on the RHS of (\ref{18a}),
   in the classical continuous cases, and also classical discrete cases. This is equivalent to studying
  \begin{equation}\label{18b} 
 \int_I e^{t x} \Big ( \rho_{(1),N}(x) -  \rho_{(1),N-1}(x)  \Big ) \, dx,
  \end{equation} 
  which is the exponential generating function for the difference $m_{k,N} - m_{k,N-1}$. 
  Making essential use of the Pearson equation, Ledoux derived a fourth order differential (difference) equation for
   $ \int_I e^{tx} (p_j(x))^2  \, d\mu$ in the classical continuous (discrete) cases respectively. 
   This methodology will be revised in Section \ref{S3}. In Section \ref{S4}, with the meaning of
  the integral, and also the exponential function $e^{tx}$,
   appropriately adjusted to the $q$ setting as in Section \ref{S2}, the classical $q$ analogues of Ledoux's integrability results are obtained.
  Thus,  the $q$ case as described in Section \ref{1.2} is missed from consideration in \cite{ledoux04,ledoux05}. 
   Specifically, in Section \ref{S4} the approach of  \cite{ledoux04} is used to derive a fourth order $q$-differential equation
   satisfied by the $q$-Laplace transform of the un-normalised measure $(p_n(x;q))^2 \, d_q \mu$.
   Here $d_q \mu$ denotes the appropriately weighted $q$-lattice corresponding to classical
$q$-orthogonal polynomials $\{ p_n(x;q) \}$. Crucial to our considerations is the $q$-Pearson equation (\ref{16b}).
   
   In \cite{ledoux04,ledoux05}, knowledge that (\ref{18b}) satisfies a fourth order differential (difference) equation was not
   used in obtaining explicit formulas for $\{ m_{k,N} \}_{k=1}^\infty$ in the classical cases. Instead, for the  classical continuous
  weights, use was made of the fact that the derivative with respect to $t$ of the LHS of (\ref{18a}) is then simply
  related to
   \begin{equation}\label{18c} 
   \int_I e^{tx} p_N(x)  p_{N-1}(x) \, d\mu .
    \end{equation} 
  The methods used in \cite{ledoux04} to study (\ref{18b}) applied to (\ref{18c}) then lead to a second order linear recurrence for the
  even moments in the case of the classical Gaussian weight (obtained first in \cite{HZ86}), and a second order linear
  recurrence for the moments in the case of the classical Laguerre weight (obtained first in \cite{HT03}).
  In the case of classical discrete weights on a linear lattice, there is no special property associated with the discrete (or continuous)
  derivative of the LHS of (\ref{18a}). Instead in \cite{ledoux05} special properties of the shifted $k$-th moment
    \begin{equation}\label{18d}  
   \int x (x -1) \cdots (x - k + 1)   (p_j(x))^2  \, d\mu 
   \end{equation} 
   for the Charlier and Meixner classical discrete weights, together with recurrences used in the derivation of 
   the fourth order difference equation relating to (\ref{18b}) as applies for all the discrete classical weights, were
   used to obtain formulas in these special cases. This working was extended to the Krawtchouk weight in the recent
   work \cite{cohen20}, and moreover the evaluations have been identified in terms of hypergeometric polynomials.
   
With these points in mind,
      in Section \ref{S5} use is made of the   $q$-Pearson equation  to study the moments of the measure
   $(H_N(x;q))^2 d_q \mu$, where $H_N(x;q)$ denotes the normalised discrete $q$-Hermite polynomials and $d_q \mu$
   denotes the 
   $q$-Hermite weight supported on the appropriate $q$-lattice, studied from the viewpoint of
   the Schur average (\ref{17a}) in Section \ref{S2}.  This is of particular interest for its combinatorial meaning
   in terms of rook placements \cite{Wi12}.


\section{Averages over Schur polynomials}\label{S2}

\subsection{The Stieltjes-Wigert weight (\ref{15f})}\label{S2.1}
Let $(\widetilde{SW})$ refer to an average with respect to the weight (\ref{15f}) in (\ref{1.1}) on the $q$-lattice in the sense of (\ref{15d}).
Let (SW) refer to an average with respect to the weight (\ref{15b}) in (\ref{1.1}). Our interest is in the moments of the density,
$ m_{k,N}^{(\widetilde{SW})}$, which we calculate via the identity  (\ref{13b+}).

With $s_\kappa$ the Schur polynomial (\ref{17b}),
according to (\ref{15b+}) and (\ref{15g}) we have that  $\langle s_\kappa \rangle^{(\widetilde{SW})} = \langle s_\kappa \rangle^{({SW})}$.
The exact form of the average on the RHS is known from \cite[Prop.~3.1]{forrester202},
\cite[Appendix A.2]{santilli21} as well as the earlier work \cite{DT07}.
 Reading off from these references then gives
  \begin{equation}\label{20a}
 \langle s_\kappa \rangle^{(\widetilde{SW})}  =
  q^{- {1 \over 2} \sum_{l=1}^N \kappa_l^2}
 \prod_{1 \le j < k \le N} {1 - q^{- (\kappa_j - j - \kappa_k + k)} \over 1 -  q^{-( k - j )}}.
 \end{equation} 
 As shown in   \cite[Coroll.~3.3 and Prop.~1.1]{forrester202}, specialising to $\kappa = (k-r,1^r)$ and making use of
 (\ref{13b+}) then implies that $(q^{Nk}/N) m_{k,N}^{(\widetilde{SW})}$ is given by (\ref{15c}), or equivalently
  \begin{equation}\label{20b} 
{1 \over N} q^{(N -1/2) k} m_{k,N}^{(\widetilde{SW})} 
  =  - {1 \over N} {( - q^{-1/2} )^k \over 1 - q^{-k}}  p_k^{(lq\text{-}J)}(q^{-N};1,q|q^{-1}).
 \end{equation} 
 Here $p_k^{(lq\text{-}J)}$ denotes the little $q$-Jacobi polynomial, specified in terms of 
 the $q$-hypergeometric function (\ref{14a+}) by
   \begin{equation}\label{20c} 
 p_n^{(lq\text{-}J)}(x;a,b|q) = {}_2 \phi_1  \Big ( {q^{-n}, a b q^{n+1} \atop a q} \Big | q; q x \Big ).
   \end{equation}   
   
 There is a further equivalent form of $m_{k,N}^{(\widetilde{SW})}$.
  This involves the $q$-binomial coefficient
    \begin{equation}\label{24b} 
  \bigg [ \begin{array}{cc} n \\ l \end{array} \bigg ]_q = { (1 - q^n) (1 - q^{n-1}) \cdots (1 - q^{n-l+1}) \over (1 - q^l) (1 - q^{l-1}) \cdots (1 - q)}
 \end{equation}  
 for $n, l$ non-negative integers.
 Thus from \cite[Eq.~(3.15) combined with (3.14)]{forrester202} we read off
  \begin{align}\label{20b+} 
 q^{(N-1/2)k} m_{k,N}^{(\widetilde{SW})} 
  & = \sum_{r=0}^{k-1} (-1)^r 
    q^{- (k-r)^2/2 - r/2}   \bigg [ {N + k - r - 1 \atop k} \bigg ]_{q^{-1}}
   \bigg [ {k- 1 \atop r} \bigg ]_{q^{-1}} \nonumber \\
   & = q^{-k N - (k^2 - 2k)/2}
  \sum_{r=0}^{k-1} (-1)^r 
    q^{ (r^2 + r)/2 + kr}   \bigg [ {N + k - r - 1 \atop k} \bigg ]_{q}
   \bigg [ {k- 1 \atop r} \bigg ]_{q},  
 \end{align}   
 where the final equality follows from the general property of the $q$-binomial
 coefficients
    \begin{equation}\label{24c} 
  \bigg [ \begin{array}{cc} n \\ l \end{array} \bigg ]_{q^{-1}} = q^{-(n-l)l}      \bigg [ \begin{array}{cc} n \\ l \end{array} \bigg ]_{q}.
  \end{equation}
  The utility of (\ref{20b+}) is that it allows for the computation of the generating function with respect to $N$ (cf.~\cite[Eq.~(2.56)]{CZ21}).
  
  \begin{proposition}\label{p2.0}
  Introduce the generating function
    \begin{equation}\label{33e} 
    G_k^{(\widetilde{SW})}(z) = \sum_{N=1}^\infty (q^{2k} z)^N     m_{k,N}^{(\widetilde{SW})}.
   \end{equation}    
  For integers $k \ge 1$ we have 
   \begin{equation}\label{34}  
  G_k^{(\widetilde{SW})}(z) =   q^{-(k^2 - 2k)/2}  {z \over 1 - z}   { (q^{k+1}z ;q)_{k-1}  \over (qz;q)_{k}}. 
  \end{equation}  
  \end{proposition}

   \begin{proof}
Simple manipulation gives
 \begin{equation}\label{ts.1}
 \sum_{N=1}^\infty z^N  \left [ \begin{array}{cc} N + k - r - 1 \\ k \end{array} \right ]_q =
 \sum_{N=1}^\infty  z^N   \left [  \begin{array}{cc} N  +k - r - 1 \\ N - r - 1 \end{array} \right ]_q  = z^{r+1}  \sum_{N=0}^\infty z^N   \left [  \begin{array}{cc} N  + k \\ N \end{array} \right ]_q.
  \end{equation}  
This last sum can be evaluated \cite[Eq.~17.2.38]{DLMF} to give
 \begin{equation}\label{ts.2}
{z^{r+1} \over (z;q)_{k+1}}.
  \end{equation}  
Thus
\begin{align*}
G_{k}^{(\widetilde{SW})}(z) & = { z  q^{-(k^2 - 2k)/2}   \over (z;q)_{k+1}} \sum_{r=0}^{k-1} 
(-z)^r 
    q^{ (r^2 + r)/2 + kr} 
   \bigg [ {k- 1 \atop r} \bigg ]_{q}. 
   \end{align*}
Performing the sum according to the
the $q$ generalisation of the simple binomial expansion  \cite[Eq.~17.2.35]{DLMF} gives (\ref{34}).

\end{proof}

We will see that very similar working to that in the above proof can be used to compute the generating function
with respect to $N$ for the moments corresponding to (\ref{1.1}) with the discrete $q$-Hermite weight to be
considered in the next subsection. If we start with (\ref{20b}) instead of (\ref{20b+}) a functional form
distinct from (\ref{34}) results.

  \begin{proposition}\label{p2.0+}
  As an alternative to the evaluation (\ref{34}) we have
   \begin{equation}\label{34+}  
  G_k^{(\widetilde{SW})}(z) =   z \sum_{s=0}^k b_s {q^s \over 1 - z q^s},
  \end{equation}  
  where
     \begin{equation}\label{34.1} 
     b_s = (-1)^s  q^{- (1/2)k^2 + k}  q^{s(s+1)/2} {(q;q)_{2k-s-1} \over (q;q)_s ((q;q)_{k-s})^2}.
   \end{equation} 
       \end{proposition} 

\begin{proof}
Starting from (\ref{20b}), simple manipulation shows
 \begin{equation}\label{34.2} 
 q^{2Nk} m_{k,N}^{(\widetilde{SW})}  = \sum_{s=0}^k b_s q^{Ns},
  \end{equation} 
 where $b_s$ is given by (\ref{34.1}). Substituting in the definition (\ref{33e}), the
 result (\ref{34+}) follows by performing a geometric series.
 \end{proof}
 
 Notice that (\ref{34+}) corresponds to the partial fractions expansion of (\ref{34}).
 Also, in keeping with the large $N$ expansion
  \begin{equation}\label{23a+}
  {1 \over N} m_{k,N}^{(\widetilde{SW})}(q) \Big |_{q = e^{-\lambda/N}} =
  \mu_{k,0}^{(\widetilde{SW})}(\lambda) + {1 \over N^2}   \mu_{k,2}^{(\widetilde{SW})}(\lambda)  +
 {1 \over N^4}   \mu_{k,4}^{(\widetilde{SW})}(\lambda)  + \cdots
 \end{equation} 
 where
    \begin{equation}\label{23b+}
     \mu_{k,0}^{(\widetilde{SW})}(\lambda)  = {(-1)^k \over \lambda k} \, {}_2 F_1(-k,k,1;e^\lambda),
 \end{equation} 
 as remarked upon in the text below (\ref{15c}), we can check from (\ref{34.1}) that
 $b_s |_{q \mapsto q^{-1}} = - b_s$ and furthermore that $(1/N) b_s |_{q = e^{-\lambda/N}}$ has
 a well defined limit. We remark too that with the scaled density
   \begin{equation}\label{23c+}
  {\rho}_{(1),0}^{(\widetilde{SW})}(x;\lambda) := \lim_{N \to \infty} {1 \over N} \rho_{(1),N}^{(\widetilde{SW})}(x) \Big |_{q = e^{-\lambda/N}},
  \end{equation} 
  characterised by its relationship to $  \mu_{k,0}^{(\widetilde{SW})}(\lambda)$ in (\ref{23a}),
    \begin{equation}\label{23d+}
   \mu_{k,0}^{(\widetilde{SW})}(\lambda)  = \int_0^\infty \lambda^k       {\rho}_{(1),0}^{(\widetilde{SW})}(x;\lambda)  \, d x,
   \end{equation} 
   it was shown in \cite{forrester202} how knowledge of the explicit functional form (\ref{23b+}) can be used
   to deduce that
     \begin{equation}\label{23e+}
   {\rho}_{(1),0}^{(\widetilde{SW})}(x;\lambda)   =  {1 \over \pi \lambda x} \arctan \bigg ( \sqrt{4 e^\lambda x - (1 + x)^2 \over 1 + x} \bigg )
    \chi_{z_- < x < z_+},
 \end{equation} 
 where    $z_\pm = - z \pm ( z^2 - 1)^{1/2}$ with $z = 1 - 2 e^\lambda$.
 

\subsection{The discrete $q$-Hermite weight}\label{S2.2}
Making use of the notation used in (\ref{15e}), the discrete $q$-Hermite weight is
\begin{equation}\label{24}
 w^{(d\text{-}qH)}(x)  = { ( q x, - qx; q)_\infty \over (q,-1,-q;q)_\infty}.
\end{equation}
This is supported on the $q$-lattice corresponding to (\ref{16d}) with $a=-1$ and $b=1$. Let an average
with respect (\ref{1.1}) defined by the Jackson integral (\ref{16d}) in each variable with $a=-1$ and $b=1$
and with weight (\ref{24}) be denoted $\langle \cdot \rangle^{(d\text{-}qH)}$.

\begin{proposition}\label{p2.1}
Consider a partition $\kappa = (\kappa_1,\dots, \kappa_N)$. Require that
the number of $j$ such that $\kappa_{j} + N - j $ is even must equal
$[(N+1)/2]$, while the number of   $j$ such that $\kappa_{j} + N -j$ is odd must equal
$[N/2]$. 
Let these $j$ be denoted $\{ j_l^{\rm e} \}_{l=1}^{[(N+1)/2]}$ and
$\{ j_l^{\rm o} \}_{l=1}^{[N/2]}$ respectively, and ordered from smallest to biggest
within each set. 
In the case of the empty partition, denote these same $j$ by 
$\{ j_l^{\rm e}(0) = 2l-1 \}_{l=1}^{[(N+1)/2]}$ and
$\{ j_l^{\rm o}(0) = 2l \}_{l=1}^{[N/2]}$. Let $S_0$ equal the number of
$\kappa_j$ that are odd.
We have
\begin{multline}\label{224a}
\langle s_\kappa \rangle^{(d\text{-}qH)} = (-1)^{S_0/2} 
\prod_{l=1}^{[(N+1)/2]} 
{(q;q^2)_{(\kappa_{ j_l^{\rm e} } +  N - j_l^{\rm e} )/2} 
\over (q;q^2)_{ ( N - j_l^{\rm e}(0) )/2} }
 \prod_{1 \le k < l \le [(N+1)/2]}
{ q^{\kappa_{j_k^{\rm e} } +  N -  j_k^{\rm e} } -  q^{\kappa_{ j_l^{\rm e} } + N -   j_l^{\rm e} } \over
 q^{N -  j_k^{\rm e}(0) } -  q^{N- j_l^{ \rm e}(0) } } \\
 \times
 \prod_{l=1}^{[N/2]} 
{(q;q^2)_{(\kappa_{ j_l^{\rm o} } + N + 1 -  j_l^{\rm o})/2}  
\over (q;q^2)_{ (N + 1 - j_l^{\rm o}(0) )/2} }
\prod_{1 \le k < l \le [N/2]}
{ q^{\kappa_{j_k^{\rm o} } + N -  j_k^{\rm o} } -  q^{\kappa_{ j_l^{\rm o} } +N -   j_l^{\rm o} } \over
 q^{ N - j_k^{\rm o}(0) } -  q^{ N -  j_l^{\rm o}(0) } }.
\end{multline}
If the requirement relating to the parity of $\{\kappa_{j} + N - j  \}$ does not hold,
then $\langle s_\kappa \rangle^{(d\text{-}qH)} = 0$. 
\end{proposition}

\begin{proof}
Denote
  \begin{equation}\label{21a} 
  I_\kappa^{(d\text{-}qH)} := {1 \over N!} \int_{-1}^1 d_q x_1 \, w^{(d\text{-}qH)}(x_1) \cdots  \int_{-1}^1 d_q x_N \, w^{(d\text{-}qH)}(x_N) \, s_\kappa(x_1,\dots,x_N)
  \prod_{1 \le j < k \le N} (x_k - x_j)^2.
  \end{equation}  
  Making use of the Vandermonde determinant identity
 \begin{equation}\label{21b}   
 \det [x_j^{k-1} ]_{j,k=1}^N = \prod_{1 \le j < k \le N} (x_k - x_j)
   \end{equation}  
   to substitute for the product of differences in (\ref{21a}), and of (\ref{17b}) to substitute for the Schur polynomial, we see
 \begin{align}\label{21c} 
  I_\kappa^{(d\text{-}qH)}  & = {1 \over N!}  \int_{-1}^1 d_q x_1 \, w^{(d\text{-}qH)}(x_1) \cdots  \int_{-1}^1 d_q x_N \, w^{(d\text{-}qH)}(x_N) \,
  \det [ x_j^{\kappa_{N+1-k} + k - 1} ]_{j,k=1}^N \det [ x_j^{k-1} ]_{j,k=1}^N  \nonumber \\
  & = \det \Big [  \int_{-1}^1   w^{(d\text{-}qH)}(x)  x^{\kappa_{N+1-j} + j + k - 2} \, d_q x \Big ]_{j,k=1}^N,
  \end{align}
  where the second equality follows from Andr\'eief's identity \cite{An86,Fo18}.
  
  The final determinant in (\ref{21c}) consists of the moments of the discrete $q$-Hermite weight, which we know are
  given by (see e.g.~\cite{NS13})
   \begin{equation}\label{21d}  
  \int_{-1}^1   w^{(d\text{-}qH)}(x)  x^{k}     \, d_q x = (1 - q) \Big ( {1 + (-1)^k \over 2} \Big ) (q;q^2)_{k/2}.
  \end{equation}
  In particular, this tells us that along each row of the determinant every second entry is zero. Interchange of
  rows and columns can therefore bring the determinant to the block form
    \begin{equation}\label{21e}  
    \det \begin{bmatrix} A_1 & 0 \\ 0 & A_2 \end{bmatrix},
 \end{equation}
 where the rows in $A_1$ have $\kappa_{N+1-j} + j -1$ even, while the rows in
 $A_2$ have $\kappa_{N+1-j} + j -1$  odd. Furthermore, the number of columns in
 $A_1$ is $[(N+1)/2]$, while the number of columns in $A_2$ is $[N/2]$.
 For (\ref{21e}) to be nonzero we require    $A_1$ and $A_2$ to be square. After replacing
 $j$ by $N+1-j$ it follows that
 the number of $j$ such that $\kappa_{j} + N - j $ is even (odd) must equal
$[(N+1)/2]$ ($[N/2]$).  Let these $j$ be denoted as in the statement of the proposition.

The block form of the determinant (\ref{21e}) implies the factorisation
  \begin{equation}\label{22a}  
    I_\kappa^{(d\text{-}qH)} = (-1)^{S_0/2} (1 - q)^N \det A_1 \det A_2
 \end{equation}
 where
  \begin{align}\label{22b}    
  A_1 & = \Big [ (q; q^2)_{(\kappa_{ j_l^{\rm e} } + N -  j_l^{\rm e}  + 2(k-1))/2 } \Big ]_{j,k=1}^{[(N+1)/2]} \nonumber \\
    A_2 & = \Big [ (q; q^2)_{(\kappa_{ j_l^{\rm o} } +  N + 1 - j_l^{\rm o}  + 2(k-1))/2 } \Big ]_{j,k=1}^{[N/2]},
  \end{align}
 and $S_0$ denotes the number of parts $\kappa_j$ that are odd.
Note that  the factor $(-1)^{S_0/2}$ is due to it being necessary to undertake $S_0/2$ row interchanges.
  
  Now, for general $a_j \in \mathbb Z_{\ge 0}$
   \begin{equation}\label{23a}  
   (q;q^2)_{a_j + k - 1} = (q;q^2)_{a_j} (q^{1+2a_j}; q^2)_{k-1}
  \end{equation}
  and thus
     \begin{equation}\label{23b}  
     \det \Big [ (q;q^2)_{a_j + k - 1} \Big ]_{j,k=1}^n = \prod_{j=1}^n (q;q^2)_{a_j}
     \det \Big [ ( q^{1 + 2a_j}; q^2)_{k-1} \Big ]_{j,k=1}^n.
 \end{equation}
 The determinant on the RHS has the known evaluation \cite{No15}     
  \begin{equation}\label{23c} 
\prod_{1 \le k < l \le n} (q^{1 + 2 a_k} - q^{1 + 2 a_l}) \prod_{k=0}^{n-1} q^{2k ( n - 1 - k)}.
  \end{equation}
  It follows that both determinants in (\ref{22b}) can be evaluated explicitly.
  Doing this, and normalising to unity for the empty partition, 
  (\ref{224a}) is obtained.

  \end{proof}

  From the definition of the Jackson integral (\ref{16d}), for any $a > 0$
   \begin{equation}\label{A.1}
   \int_{-1}^1 f(x) \, d_q x =  {1 \over a} \int_{-a}^a f \Big ( {x \over a} \Big ) \, d_q x
   \end{equation}
   as is familiar in the continuous case. Moreover, for the discrete Hermite weight (\ref{24})
   it is well known that (see \cite{AAT19} for details of the required working) 
     \begin{equation}\label{A.1a} 
     \lim_{q \to 1^-} {1 \over a}  w^{(d\text{-}qH)}\Big ( {x \over a} \Big ) \Big |_{a = (1 - q^2)^{-1/2}} =
     {1 \over \sqrt{\pi}} e^{- x^2}.
   \end{equation}   
   The significance of (\ref{A.1}) and (\ref{A.1a}) in the context of averages
   $\langle \cdot \rangle^{(d\text{-}qH)}$ is that they imply
   \begin{equation}\label{A.1b}  
       \lim_{q \to 1^-}   (1 - q^2)^{-|\kappa|/2}   \langle s_\kappa \rangle^{(d\text{-}qH)} =   \langle s_\kappa  \rangle^{(G)},
   \end{equation}          
   where the average on the RHS is with respect to (\ref{1.1}) with the Gaussian weight $w^{(G)}(x)$ from
   (\ref{1.1b}). Applying (\ref{A.1b}) then allows for the computation of $ \langle s_\kappa \rangle^{(G)}$.
   
   \begin{coro}
   Let the parity of the parts of $\kappa$ be restricted as in Proposition \ref{p2.1}, and furthermore make use of the
   notation therein.  We have
 \begin{multline}\label{34a}
 \langle 
s_\kappa \rangle^{(G)} = (-1)^{S_0/2} 
 \prod_{l=1}^{[(N+1)/2]} 
 { (\kappa_{ j_l^{\rm e} } +  N - j_l^{\rm e} - 1)!!
 \over ( N - j_l^{\rm e}(0) - 1)!!}
  \prod_{1 \le k < l \le [(N+1)/2]}
{ \kappa_{ j_l}^{\rm e}  -   j_l^{\rm e}   - \kappa_{j_k}^{\rm e}  +  j_k^{\rm e}  \over 
  j_l^{\rm e} -   j_l^{\rm e} } \\ \times
   \prod_{l=1}^{[N/2]} 
{ (\kappa_{ j_l^{\rm o} } + N  -  j_l^{\rm o}) !! \over
 (N  - j_l^{\rm o}(0) )!! }
\prod_{1 \le k < l \le [N/2]}
{ \kappa_{ j_l}^{\rm o}  -   j_l^{\rm o}   - \kappa_{j_k}^{\rm o}  +  j_k^{\rm o}  \over 
  j_l^{\rm o} -   j_l^{\rm o} }.
\end{multline}  
 \end{coro}  
 
 \begin{remark}\label{Rm1}
 The evaluation of $\langle s_\kappa \rangle^{(G)}$, albeit in a different but equivalent form to that on the
 RHS of (\ref{34a}), was first obtained by Di Francesco and Itzykson \cite{DI93}. Their derivation
 made essential use of the character interpretation of the Schur polynomials. Subsequently, but without
 the details being given, it was noted in \cite{KSW96} that a direct derivation along the lines of that implied in
 the proof of Proposition \ref{p2.1} could be given.
 \end{remark}

  We now specialise (\ref{34a}) to partitions of the form $\kappa = (2k-r,1^r)$.

  \begin{coro} (Venkataramana \cite[Th.~4, with the leading sign and exponent of the
  factor of $q$ corrected]{PV14})
For $r$ a non-negative integer, $0 \le r \le 2k$, we have
 \begin{equation}\label{24a}
\langle s_{(2k-r,(1)^r)} \rangle^{(d\text{-}qH)} = (-1)^{\floor{(r+1)/2}} q^{\floor{r/2}(\floor{r/2} + 1)}
  \left [ \begin{array}{cc} N + 2k - r - 1 \\ 2k \end{array} \right ]_q
   \left [ \begin{array}{cc}  k  - 1 \\  \floor{r/2} \end{array} \right ]_{q^2} (q; q^2)_k.
   \end{equation}
\end{coro}

\begin{proof}
Our task is to simplify (\ref{34a}) in the case of $\kappa_1 = 2k - r$, $\kappa_2 = 1, \dots, \kappa_{r+1}=1$.
For convenience we suppose that $N$ is even; the working in the case $N$ odd is analogous, and leads
to the same result.

To begin, we observe that the contribution from the single products in (\ref{24a}), the first involving
$\{ j_l^{\rm e} \}$ and the second involving $\{ j_l^{\rm o} \}$, comes from the first (second) product for
$r$ odd (even). This contribution is seen to equal
 \begin{equation}\label{30a}
 {(q;q^2)_{N/2 + k -  \floor{r} } \over
 (q;q^2)_{ N/2 - \floor{r} } }.
  \end{equation}
  
  It is similarly true that the contribution from the double products in (\ref{24a}), the first involving
$\{ j_l^{\rm e} \}$ and the second involving $\{ j_l^{\rm o} \}$, comes from the first (second) product for
$r$ odd (even). After simplification, this contribution is seen to equal
 \begin{equation}\label{30b}
 q^{\floor{r/2}(\floor{r/2} + 1)}
    \left [ \begin{array}{cc}  k  - 1 \\  \floor{r/2} \end{array} \right ]_{q^2} 
 {(q^2;q^2)_{N/2 + k - \floor{r/2} - 1} \over (q^2;q^2)_{N/2  - \floor{r/2} - 1}}  {(q; q^2)_k \over   (q^2;q^2)_k }.
 \end{equation} 
 According to   (\ref{34a}) there is also a factor
  \begin{equation}\label{30c}
  (-1)^{ \floor{(r+1)/2} } .
  \end{equation} 
  
  Multiplying together (\ref{30a}),  (\ref{30b}), (\ref{30c}) and simplifying gives (\ref{24a}).

 \end{proof}

 \begin{remark} The proof of (\ref{24a}) given in \cite{PV14} is different to the one above, relying on 
 properties of particular multivariate discrete $q$-Hermite polynomials; for generalisations of the
 latter in the setting of Macdonald polynomial theory involving parameter $(q,t)$, see \cite{BF00}. The average of a Macdonald
 polynomial in the latter setting has been conjectured in \cite{LPSZ20}. The case $q=t$ corresponds to the Schur
 average in Proposition \ref{p2.1}, giving a character based formula which is a $q$-generalisation of the evaluation
 of (\ref{34a}) in \cite{DI93}.  In the recent work \cite{MPS20}, specialising the latter to partitions $\kappa = (2k-r, (1)^r)$ has been
 shown to give a formula equivalent to (\ref{24a}).
 \end{remark}

 Substituting (\ref{24a}) in (\ref{13b+}) gives for the moments of the density of the discrete $q$-Hermite ensemble \cite{PV14}
   \begin{equation}\label{30d}
   m_{2k,N}^{(d\text{-}qH)} = \sum_{r=0}^{2k-1}  (-1)^{r+\floor{(r+1)/2}} q^{\floor{r/2}(\floor{r/2} + 1)}
  \left [ \begin{array}{cc} N + 2k - r - 1 \\ 2k \end{array} \right ]_q
   \left [ \begin{array}{cc}  k  - 1 \\  \floor{r/2} \end{array} \right ]_{q^2} (q; q^2)_k;
   \end{equation}
cf.~(\ref{20b+}). For example, with $k=1$
 \begin{align}\label{30e}
   m_{2,N}^{(d\text{-}qH)} & =  (q;q^2)_1 \bigg ( \begin{bmatrix} N + 1 \\ 2 \end{bmatrix}_q +  \begin{bmatrix} N  \\ 2 \end{bmatrix}_q  \bigg ) = {1 \over 1 - q^2} \Big (
   q^{2N} (q + q^{-1}) - q^N (q + 2 + q^{-1}) + 2 \Big ). 
  \end{align}

It has recently been observed \cite{MPS20} that (\ref{30d})  simplifies upon forming the generating
function with respect to $N$.

\begin{proposition} (Morozov, Popolitov and Shakirov \cite[Eq.~(4-7) with $\lambda = z q^{2m}, m = k$]{MPS20})   
Define
 \begin{equation}\label{e.1}
G_{2k}^{(d\text{-}qH)}(z) = \sum_{N=1}^\infty z^N  m_{2k,N}^{(d\text{-}qH))}.
 \end{equation}
For integers $k \ge 1$ we have
   \begin{equation}\label{30e+}
  G_{2k}^{(d\text{-}qH)}(z) =   {z (1 + z) \over 1 - z} { ((qz)^2;q^2)_{k-1} \over  (qz;q)_{2k} } (q; q^2)_k.
 \end{equation}
 \end{proposition}   
 
 \begin{proof}
 This can be established by working analogous to that used in the derivation of Proposition \ref{p2.0}.
 Thus, according to (\ref{ts.1}) and (\ref{ts.2}) with $k$ replaced by $2k$ we have
 $$
 \sum_{N=1}^\infty z^N  \left [ \begin{array}{cc} N + 2k - r - 1 \\ 2k \end{array} \right ]_q =
{z^{r+1} \over (z;q)_{2k+1}}.
$$
Hence
\begin{align*}
G_{2k}^{(d\text{-}qH)}(z) & = {(q; q^2)_k  \over (z;q)_{2k+1}} \sum_{r=0}^{2k-1} z^{r+1} (-1)^{r+\floor{(r+1)/2}} q^{\floor{r/2}(\floor{r/2} + 1)}
   \left [ \begin{array}{cc}  k  - 1 \\  \floor{r/2} \end{array} \right ]_{q^2}  \\
 & =  z (1 + z) \sum_{r=0}^{k-1} z^{2r} (-1)^{r} q^{r(r+1)}
   \left [ \begin{array}{cc}  k  - 1 \\  r \end{array} \right ]_{q^2}, 
   \end{align*}
where the second equality follows by breaking up the sum in the first equality according to the parity of $r$.  
But according to the $q$ generalisation of the simple binomial expansion  \cite[Eq.~17.2.35]{DLMF}, this last sum has the evaluation
$((qz)^2;q^2)_{k-1}$, and (\ref{30e+}) follows.

\end{proof}

In keeping with the relationship between (\ref{34.2}), (\ref{34+}) and (\ref{33e}) we can use (\ref{30e+}) to deduce the coefficients in the expansion
   \begin{equation}\label{F.1}
   q^k m_{2k,N}^{(d\text{-}qH)} = \sum_{p=0}^{2k} c_p q^{pN}
  \end{equation}  
  (here the factor of $q^k$ in the LHS is for later convenience).
  
  \begin{proposition}\label{PM} (Morozov, Popolitov and Shakirov \cite[Equivalent to Eqns.~(3-7) and (3-8)]{MPS20})   
  We have
  \begin{align*}
  c_0 & = {2 q^k \over 1 - q^{2k}} \\
  c_1 & = \left \{ \begin{array}{ll} 0, & k > 1 \\
  - {1 + q \over 1 -  q}, & k = 1. \end{array} \right.
  \end{align*}
  For $p \ge 2$ we have
   \begin{align*}
  c_p &= 0, \quad   2 \le p < k \\
  c_p  & = (-1)^{k+p-1} \Big ( {1 + q^p \over 1 - q^p} \Big ) q^{p(p+1)/2 - 2 p k + k (k-1)}
  {(q^2;q^2)_{p-1} (q;q^2)_k \over
  (q^2;q^2)_{p-k} (q;q)_{p-1} (q;q)_{2k-p}}, \quad k \le p \le 2k.
  \end{align*}
  \end{proposition}

   \begin{proof}
   Substituting (\ref{F.1}) in (\ref{e.1}) we have
 $$
 q^k G_{2k}^{(d\text{-}qH)}(z) =  z \sum_{p=0}^{2k} {c_p q^p \over 1 - z q^p}.
 $$
 Hence, with ${\rm Res}_{z = z_0} f(z)$ denoting the residue at $z = z_0$ of $f(z)$,
 $$
 c_p = - q^{p+k}  \, {\rm Res}_{z = q^{-p}} \, G_{2k}^{(d\text{-}qH)}(z).
 $$
 The product form of $G_{2k}^{(d\text{-}qH)}(z)$ (\ref{30e+}) allows for a simple computation of
 the residues, and the stated results follow.
 \end{proof}
 
 It is simple to check from the results of Proposition \ref{PM} that $c_p |_{q \mapsto q^{-1}}= - c_p$.  We can also check
 that  with
 $q = e^{- \lambda/N}$, dividing $c_p$ by $N$ leads to a well defined limit. Thus analogous to (\ref{14c+}) the scaled moments
 possess a $1/N^2$ expansion
 \begin{equation}\label{14c+G}   
 {1 \over N}  m_{k,N}^{(d\text{-}qH))}(q) \Big |_{q = e^{-\lambda/N}} =  \mu_{k,0}^{(d\text{-}qH))}(\lambda) + {1 \over N^2}  \mu_{k,2}^{(d\text{-}qH))}(\lambda) + {1 \over N^4}
  \mu_{k,4}^{(d\text{-}qH))}(\lambda) + \cdots.
   \end{equation} 
 Moreover, with $x= q^{-\lambda}$, we  read off from the results of Proposition   \ref{PM} that
  \begin{equation}\label{14c+H}   
  \lambda \mu_{k,0}(\lambda) = {2 \over k} - \delta_{k,1} x + (-1)^{k-1}
  \sum_{p=k \atop p \ne 1}^{2k} (-1)^p {2 \over p} {(2(p-1))!! (2k - 1)!! \over (2(p-k))!! (p-1)! (2k - p)!} x^p.
  \end{equation}
  However, it is not apparent if knowledge of this explicit expression can used to compute the corresponding
  scaled density analogous to (\ref{23c+})--(\ref{23e+}) in the Stieltjes-Wigert case.
  
  \begin{remark}
  It follows from (\ref{A.1b}) that
    \begin{equation}\label{m.0}
         \lim_{q \to 1^-}   (1 - q^2)^{-|\kappa|/2}   m_{2k,N}^{(d\text{-}qH)} =    m_{2k,N}^{(G)},
         \end{equation}
         where $ m_{2k,N}^{(G)}$ refers to the moments of the density for the Gaussian case
         of (\ref{1.1}). The form (\ref{F.1}) is not suitable for taking this limit, as the limit
         does not exist term-by-term. It is possible to set 
 $q=1$   in each term of (\ref{30d}).  Considering separately the $r$ odd and $r$ even terms shows
  \begin{equation}\label{m.1}
m_{2k,N}^{(G)} = (2k - 1)!! \sum_{r=0}^{k-1} (-1)^r \bigg \{
  \binom{N+2k-2r-1}{2k} +   \binom{N+2k-2r-2}{2k} \bigg \} \binom{k-1}{r}.
  \end{equation} 
  Another formula for which it is possible to directly set $q=1$ is 
  (\ref{30e+}). Doing this we see \cite[Eq.~(3-4)]{MPS20}
   \begin{equation}\label{m.2} 
  G_{2k}^{(d\text{-}qH)}(z) \Big |_{q=1} = (2k - 1)!!  {z (1 + z)^k \over (1 - z)^{k+2}}.
  \end{equation}
  Calculating the coefficient of $z^N$, which we do 
  this using the formula for the product of two power series, then implies
  \begin{equation}\label{m.3}  
 m_{2k,N}^{(G)} = (2k - 1)!!  \sum_{l=0}^{N-1} { (k+2)_l \over l!} \binom{k}{N-1-l}.
  \end{equation}  
  However, neither (\ref{m.1}) nor (\ref{m.3}) reveal the known special function form \cite{WF14, cunden19}
    \begin{equation}\label{m.4}
  m_{2k,N}^{(G)}  =  2^{-k} (2k - 1)!!  \, N \, {}_2 F_1(-k,1-N,2;2),
    \end{equation}     
  which moreoever is valid for continuous $k$ and satisfies the second order recurrence \cite{HZ86}
   \begin{equation}\label{m.5}
  (k+1)  m_{2k,N}^{(G)}  =  (2k-1) m_{2k - 2,N}^{(G)}  +  (k - 1/2) (k - 1) (k - 3/2) m_{2k - 4,N}^{(G)} .
  \end{equation}
  The open question along these lines is if $  m_{2k,N}^{(d\text{-}qH)} $ admits
  a $q$-special form analogous to (\ref{20b}), or satisfies a linear recurrence.
  \end{remark}
  
  \subsection{The little $q$-Laguerre weight}\label{S2.3}
  We consider next the case of (\ref{1.1}), supported on the $q$-lattice corresponding to (\ref{16d}) with $a=0$, $b=1$,
  and with the little $q$-Laguerre weight
   \begin{equation}\label{40.1}
   w^{(l\text{-}qL)}(x) = x^\alpha (qx;q)_\infty.
   \end{equation}  
   Most important for present purposes is that the moments of this weight have the simple explicit form  \cite{NS13}
 \begin{equation}\label{40.2}    
 \int_0^1    w^{(l\text{-}qL)}(x)  x^k \, d_q x = { (q;q)_\infty \over  (q^{\alpha + 1};q)_\infty } (q^{\alpha + 1}; q)_k.
  \end{equation}

  \begin{proposition}
  We have
   \begin{equation}\label{40.2}  
   \langle s_\kappa \rangle^{(l\text{-}qL)}   = \prod_{j=1}^N {(q^{\alpha+1};q)_{\kappa_j + N - j} \over (q^{\alpha+1};q)_{ N - j} }
   \prod_{1 \le j < l \le N} { q^{\kappa_j + N - j} -  q^{\kappa_l +N - l}  \over q^{N-j} - q^{N-l} }.
  \end{equation}
  \end{proposition}
  
  \begin{proof}
  Denote 
   \begin{equation}\label{21a+} 
  I_\kappa^{(l\text{-}qL)} := {1 \over N!} \int_{0}^1 d_q x_1 \, w^{(l\text{-}qL)}(x_1) \cdots  \int_{0}^1 d_q x_N \, w^{(l\text{-}qL)}(x_N) \, s_\kappa(x_1,\dots,x_N)
  \prod_{1 \le j < k \le N} (x_k - x_j)^2.
  \end{equation}  
  Analogous to (\ref{21c}) we have
 \begin{align}\label{21c+} 
  I_\kappa^{(l\text{-}qL)}  
  & = \det \Big [  \int_{0}^1   w^{(l\text{-}qL)}(x)  x^{\kappa_{N+1-j} + j + k - 2} \, d_q x \Big ]_{j,k=1}^N \nonumber \\
  & =  \Big (   { (q;q)_\infty \over  (q^{\alpha + 1};q)_\infty } \Big )^N \prod_{j=1}^N (q^{\alpha + 1}; q)_{\kappa_{N+j-1} + j - 1}
   \det \Big [  ( q^{\alpha + \kappa_{N+1-j} + j}; q)_{k-1} \Big ]_{j,k=1}^N.
  \end{align}
Use now of the determinant evaluation noted in (\ref{23c}), and normalising to unity for the empty partition, gives
(\ref{40.2}).
\end{proof}

We next make use of (\ref{13a+}) to deduce from (\ref{40.2}) the moments of the corresponding density.

\begin{proposition}\label{P2.12}
Let  the $q$-hypergeometric function ${}_3 \phi_2$ have its usual meaning.
For $k$ a positive integer we have
   \begin{equation}\label{S.4f}
  m_{k,N}^{(l\text{-}qL)}  =  A_{k,N} \,  {}_3 \phi_2  \Big ( {q^{-(k-1)}, q^{-(N-1)}, q^{-(\alpha+N - 1)}  \atop q^{- (N  + k - 1)},  q^{- (N  + k + \alpha - 1)} } \Big | q; q^{-k} \Big ),
  \end{equation} 
  where
   \begin{equation}\label{S.4e}       
 A_{k,N}  :=    { (q;q)_{N + k - 1 + \alpha}   (q;q)_{N + k - 1 }    \over
   (q;q)_{N-1} (q;q)_{k}   (q;q)_{N  +\alpha - 1} }.
  \end{equation}   
  \end{proposition}

\begin{proof}
According to (\ref{13a+}) we want to simplify (\ref{40.2}) for partitions $\kappa = (k-r,1^r)$. Noting that
\begin{equation}\label{S.1}
\prod_{1 \le j < l \le N}  { q^{\kappa_j + N - j} -  q^{\kappa_l +N - l}  \over q^{N-j} - q^{N-l} } =
q^{\sum_{l=1}^N ( l - 1) \kappa_l}
\prod_{1 \le j < l \le N} {1 - q^{\kappa_j - \kappa_l + l - j} \over 1 - q^{l-j}},
  \end{equation} 
  we see from (\ref{20a}) that the contribution from this factor in (\ref{40.2}) is, up to the prefactors, and up to
  mapping $q \mapsto q^{-1}$, the same as for the Stieltjes-Wigert weight. Thus we read off from the first equality in
  (\ref{20b+}) that the contribution of (\ref{S.1}) is
 \begin{equation}\label{S.2} 
 q^{r(r+1)/2}  \bigg [ {N + k - r - 1 \atop k} \bigg ]_{q}
   \bigg [ {k- 1 \atop r} \bigg ]_{q}  = 
  q^{r(r+1)/2}  {(q;q)_{k-1} \over   (q;q)_{k}  (q;q)_{r} } { (q;q)_{N+k-r-1} \over  (q;q)_{N-r-1} (q;q)_{k-r-1} },
    \end{equation} 
    where on the RHS we have separated terms of the form $(q;q)_{s-r-1}$ for some $s$ independent of $r$.
    For the remaining factor in (\ref{40.2}), which involves a single product, the contribution can be written
 \begin{equation}\label{S.4}     
  { (q;q)_{k-r+N-1+\alpha} \over (q;q)_{N  +\alpha - 1} } \prod_{l=1}^r (1 - q^{\alpha + N - l}) = 
  { (q;q)_{k-r+N-1+\alpha} \over (q;q)_{N  +\alpha - 1} } (-1)^r q^{r(\alpha + N) - r (r+1)/2} (q^{-(\alpha+N)+1}; q)_r.
  \end{equation} 
  
  For $r$ a non-negative integer, we have the general formula
  \begin{equation}\label{S.4a}    
  (q;q)_{m-r} = (-1)^r q^{r (r+1)/2} q^{- r (m + 1)} { (q;q)_m \over (q^{-m};q)_r},
  \end{equation} 
  which we recognise as closely related to the manipulation used to obtain the equality in (\ref{S.4}).
  This formula can be used to rewrite the terms in the product of (\ref{S.2}) and (\ref{S.4}) of the functional
  form $(q;q)_{m-r}$. Thus we have
  \begin{multline}\label{S.4b} 
   {  (q;q)_{k-r+N-1+\alpha} (q;q)_{N+k-r-1} \over  (q;q)_{N-r-1} (q;q)_{k-r-1} }  = q^{-r (k + N + \alpha)} \\
   \times { (q;q)_{N + k - 1 + \alpha}   (q;q)_{N + k - 1 }    \over
   (q;q)_{N-1} (q;q)_{k-1}}
   { (q^{-(k-1)}; q)_r (q^{-(N-1)};q)_r \over    ( q^{- (N  + k + \alpha - 1)}; q)_r   ( q^{- (N  + k - 1)}; q)_r }.
   \end{multline}
   The terms in (\ref{S.2}) and (\ref{S.4}) excluding those on the LHS of  (\ref{S.4b}) are
   \begin{equation}\label{S.4c}   
   (-1)^r q^{r (\alpha + N)} { (q;q)_{k-1}    (q^{-(\alpha+N)+1}; q)_r   \over  (q;q)_{k}   (q;q)_{r} (q;q)_{N  +\alpha - 1} }.
     \end{equation} 
     
     Multiplying together (\ref{S.4c}) and the RHS of (\ref{S.4b}) gives us a structured evaluation of
     $\langle s_{(k-r,1^r)} \rangle^{(l\text{-}qL)} $. Forming the sum as required by (\ref{13a+}),
     and specifying $A_{k,N}$ by (\ref{S.4e}), shows
   \begin{equation}\label{S.4d}     
   m_{k,N}^{(l\text{-}qL)}  = A_{k,N} \sum_{r=0}^{k-1}  {  (q^{-(k-1)}; q)_r  (q^{-(N-1)};q)_r      (q^{-(\alpha+N - 1)}; q)_r \over
      ( q^{- (N  + k - 1)}; q)_r    ( q^{- (N  + k + \alpha - 1)}; q)_r}   {q^{-r k} \over  (q;q)_{r}  }.
  \end{equation}    
 In standard $q$-hypergeometric series notation, this is (\ref{S.4f}).
   \end{proof}
   
   From the identity \cite[Eq.~17.3.2]{DLMF} and corresponding limit
   \begin{equation}\label{S.5}
   (x;q)_\infty = \sum_{n=0}^\infty { q^{n(n-1)/2} (-x)^n \over (q;q)_n }, \quad \lim_{q \to 1^-} (x(1-q);q)_\infty = e^{-x},
    \end{equation}     
    the definition (\ref{40.1}) of the little $q$-Laguerre weight, and (\ref{A.1}) with $a = 1/(1 - q)$, we see that
  \begin{equation}\label{S.5a}   
  \lim_{q \to 1^-} (1 - q)^{-|\kappa|} \langle s_\kappa \rangle^{(l\text{-}qL)}  = \langle s_\kappa \rangle^{(L)}.
   \end{equation} 
   Here the average on the RHS is with respect to (\ref{1.1}) with the Laguerre weight $w^{(L)}(x)$ from (\ref{1.1b}).
   It follows from  (\ref{13a+}) that the moments in the little $q$-Laguerre ensemble are likewise related to those in the
   Laguerre ensemble by
     \begin{equation}\label{S.5b}  
    \lim_{q \to 1^-} (1 - q)^{-|\kappa|}     m_{k,N}^{(l\text{-}qL)}  = m_{k,N}^{(L)}.
   \end{equation} 
   Computing this limit from   (\ref{S.4f}) shows
  \begin{equation}\label{S.5c} 
  m_{k,N}^{(L)} =    {(N+k-1+\alpha)! (N+k-1)! \over (N-1)! k! (N + \alpha - 1)!} \,
  {}_3 F_2 \bigg ( {-(k-1), \,- (N - 1), \,- (\alpha + N - 1) \atop - (N + k - 1), \, -(N+k+ \alpha - 1)} \Big | 1 \bigg ),
  \end{equation}
  which agrees with a known result \cite[displayed equation below (4.16)]{cunden19}.
  After transformation, (\ref{S.5c}) can be recognised as an example of a particular continuous dual Hahn
  polynomial of degree $N-1$ \cite{cunden19}, which in fact remains true for continuous $k$. In the
  case of $k$ a non-negative integer, it can be recognised as a continuous Hahn polynomial of degree
  $k-1$, and the three term recurrence of this class of orthogonal polynomial implies the known
  three term recurrence for the Laguerre ensemble moments \cite{HT03},
   \begin{equation}\label{S.5d}  
   (k+2) m_{k+1,N}^{(L)}   = (2k+1) (2N + \alpha)  m_{k,N}^{(L)}   + (k-1)(k^2 - \alpha^2)  m_{k-1,N}^{(L)}.
   \end{equation}
   In the little $q$-Laguerre case, the form    (\ref{S.4f})  does not reveal identification with a $q$-orthogonal
   polynomial, and the question as to whether $\{  m_{k,N}^{(l\text{-}qL)}  \}$ satisfies a linear recurrence remains
   open.
   
   Inspection of (\ref{S.4f}) shows the moments $m_{k,N}^{(l\text{-}qL)}$ permit the expansion
     \begin{equation}\label{S.5e}   
 m_{k,N}^{(l\text{-}qL)} = \sum_{s_1 = 0}^k    \sum_{s_2 = 0}^{2k} c_{s_1, s_2} q^{s_1 \alpha + s_2 N},
  \end{equation}
  where the coefficients $\{ c_{s_1, s_2}  \}$ are Laurent polynomials independent of $\alpha, N$.
  However, unlike the analogous expansion in the Stieltjes-Wigert and discrete $q$-Hermite cases 
  (\ref{34.2}) and (\ref{F.1}), we have no access to a general explicit evaluation of these coefficients. Equivalently,
  we are not able to find a closed form for the generating function $ \sum_{N, \alpha = 0}^\infty w^{N} z^{\alpha}
   m_{k,N}^{(l\text{-}qL)}$. For small values of $k$ (\ref{S.4f}) can be used for this purpose, and indicates that
   $q^{-1/2}   c_{s_1, s_2} $ is a series in $\{ (q^{-n/2} - q^{n/2}) \}_{n=1,2,\dots}$ and thus the scaled
   coefficients $(q^{-1/2} / N)
  c_{s_1, s_2}|_{q = e^{\lambda/N}} $ admit a $1/N^2$ expansion analogous to (\ref{23a+}) and (\ref{14c+G}).

\section{Pearson equation and Ledoux's results}\label{S3}

In this section we revise Ledoux's results in \cite{ledoux04,ledoux05} on the consequences of the Pearson equation
in the continuous and linear discrete classical cases, for the analyse of integrability properties of (\ref{18b}).
 To begin, we give a brief introduction to the Pearson equation and its consequences in random matrix theory
 and related topics, 
  and then treat the continuous case and linear discrete case separately. 

\subsection{The Pearson equation and applications}
The structured differential relation for the weight function 
\begin{align}\label{pearson}
(\sigma(x) w(x))'=\tau(x)w(x),
\end{align}
 is to be referred to as the Pearson equation in the continuous case. Here it is assumed that $\sigma(x) w(x)$ decays upon approaching
 the boundary of the support of $w(x)$.
 As is consistent with (\ref{1.1a}) --- this is the original form of the Pearson differential equation \cite{Pe95} ---  the classical weights are when  $\sigma$ and $\tau$ are polynomials
 of degrees less than 2 and 1 respectively.
 The Pearson equation has been generalised to discrete measures/nonuniform lattices; see for example \cite{ismail05,koekoek10,nikiforov86}. It was proved that for the discrete classical orthogonal polynomials, the weight function $w(s)$ should satisfy a discrete analogue of the Pearson equation
\begin{align*}
\frac{\Delta}{\Delta x(s-\frac{1}{2})}[\sigma(s)w(s)]=\tau(s)w(s),
\end{align*}
where $\Delta$ is the forward difference operator corresponding to the lattice, and $x(s)$ denotes the position of the lattice. Interesting examples include the linear discrete lattice and exponential discrete lattice. The former one induces the so-called discrete orthogonal polynomials, whose weight functions satisfy a linear discrete equation
\cite{nikiforov86}
\begin{align}\label{dpe}
\frac{w(x+1)}{w(x)}=\frac{\sigma(x)+\tau(x)}{\sigma(x+1)}.
\end{align}
and the latter one corresponds to the $q$-orthogonal polynomials, and their weight functions satisfy a $q$-difference equation (\ref{16b}).

In addition to its applications in the classification of classical orthogonal polynomials, the Pearson equation
 has seen a number of applications in the contexts of complex analysis and random matrix theory.
One class of application is to the evaluation of classical beta integrals \cite{rahman94} and Barnes and Ramanujan-type integrals on the $q$-linear lattice \cite{rahman942}. Applications in random matrix theory up to the year 2010 have been discussed in the monograph \cite{forrester10}.  The most prominent  is to connect classical orthogonal polynomials to particular
 skew orthogonal polynomials. This result was used to evaluate the Christoffel-Darboux kernel for orthogonal and symplectic ensembles with classical weights, and to show in particular that
 they are rank one perturbations of the kernel for the corresponding unitary ensemble \cite{adler00}. Subsequently the $q$-Pearson equation has been used to 
 generalise this theory, for the symplectic ensembles,  to the classical $q$-cases \cite{forrester20,SLY21}.
Differential identities that follow form the Pearson equation in the classical cases  \cite[Sec. 5.4]{forrester10} have  found application in the derivation
of formulas  for the structure function --- that is the covariance of the linear statistics $\sum_{j=1}^N \exp(ik_1\lambda_j)$, $\sum_{j=1}^N \exp(-ik_2\lambda_j)$  ---
in the Gaussian and Laguerre unitary ensembles of random matrix theory  \cite{forrester21a,forrester21b}.
 Another recent development has been in relation to the analyse of the partition function of particular  two-component log-gas systems on a line with charges $+1$ and $+2$  by making use of the Pearson pair in the classical cases \cite{forrester212}.

\subsection{Continuous case}\label{S3.2}
The continuous orthogonal polynomials $\{P_n(x)\}_{n\in\mathbb{N}}$, orthogonal with respect to the weight function $w(x)$, the Pearson pair $(\sigma,\tau)$ in \eqref{pearson} induces a second order differential equation
\begin{align}\label{de}
\sigma(x)P_N''(x)+\tau(x)P_N'(x)+\lambda_NP_N(x)=0,
\end{align}
where $\lambda_N$ is a constant independent of $x$. Since $\sigma(x) w(x)$ in the Pearson pair is required to decay fast at the ends of supports, it follows that for arbitrary $f,\,g\in L^2(d\mu,[a,b])$ there is an integration by part formula
\begin{align}\label{cibp}
\int_a^b \sigma(x) f'(x)g(x)d\mu=-\int_a^b \sigma(x) f(x)g'(x)d\mu-\int_a^b \tau(x) f(x)g(x)d\mu,
\end{align}
where $d\mu:=w(x)dx$ and $[a,b]$ is the support of $d\mu$.

The study \cite{ledoux04}  took up the task of finding
a differential equation satisfied by the (un-normalised) probability density $P_N^2(x) d\mu$, or more particularly its Laplace
transform (for earlier works relating to studies of the differential equation satisfied by products of classical orthogonal polynomials
see \cite{HBR00} and references therein).
Denote the Laplace transform of the density as
\begin{align*}
\phi(s):=\int \lp P_N^2(x)d\mu,
\end{align*}
then by making advantage of the integration by part formula \eqref{cibp} and taking $f=\lp$ and $g=P_N^2$, one can find
\begin{subequations}
\begin{align}
&s\int \lp \sigma(x) P_N^2(x)d\mu=-\int \lp \tau(x) P_N^2(x)d\mu-\int \lp \sigma(x) (P_N^2(x))'d\mu,\label{lp1}\\
&s\int\lp\sigma(x) P_N(x)P_N'(x)d\mu=-\int \lp \sigma(x) (P_N'(x))^2d\mu+\lambda_N \int\lp P_N^2(x)d\mu,\label{lp2}\\
&s\int \lp \sigma(x) (P_N'(x))^2d\mu=\int  \lp \tau(x) (P_N'(x))^2d\mu+2\lambda_N\int \lp P_N(x)P_N'(x)d\mu, \label{lp3}
\end{align}
\end{subequations}
if one take $f=e^{sx}$ and $g$ as $P_N^2(x)$, $(P_N^2(x))'$ and $(P_N'(x))^2$ respectively.

To characterize the differential equation satisfied by such Laplace transform, one can introduce a differential operator $\mathbf{g}$ corresponding to a polynomial $g(x)=\sum_{i=0}^k g_ix^i$, such that
\begin{align*}
\mathbf{g}(\phi)=\int e^{sx}g(x)P_N^2(x)d\mu.
\end{align*}
By some simple calculation, one verifies that $\mathbf{g}=\sum_{i=0}^k g_i\partial_s^i$.
Moreover, such a differential operator $\mathbf{g}$ satisfies the Leibniz rule
\begin{align}\label{oef}
\mathbf{g}(s^k\phi)=\sum_{i=0}^k {k\choose i} s^i\mathbf{g}^{(k-i)}(\phi),
\end{align}
where $\mathbf{g}^{(i)}$ is the operator corresponding to $g^{(i)}(x)$.
Therefore, one can define operators $ \mathbb{B} $ and $ \mathbb{A} $ related to $\sigma(x)$ and $\tau(x)$  via
\begin{align*}
	\mathbb{B}\phi(s)=\int \lp \sigma(x) P_N^2(x)d\mu,\quad \mathbb{A}\phi(s)=\int \lp \tau(x) P_N^2(x)d\mu,
\end{align*}
allowing equations \eqref{lp1}-\eqref{lp3} to be rewritten into the forms
\begin{subequations}
\begin{align}
	(s&\mathbb{B}+\mathbb{A})\phi =-2\mathbb{B}\int\lp P_N(x)P_N'(x)d\mu\label{eq1}\\
	s&\mathbb{B}\int\lp P_N(x)P_N'(x)d\mu=-\mathbb{B}\int\lp(P_N'(x))^2+\lambda_N\phi\label{eq2}\\
	(s&\mathbb{B}-\mathbb{A})\int\lp(P_N'(x))^2d\mu=2\lambda_N\int\lp P_N(x)P_N'(x)d\mu.\label{eq3}
\end{align}
\end{subequations}
Acting $ s\mathbb{B} $ on equation (\ref{eq3}) and substituting $ s\mathbb{B}\int\lp P_N(x)P_N'(x)d\mu $ by (\ref{eq2}) yields
\begin{align*}
	s\mathbb{B}(s\mathbb{B}-\mathbb{A})\int\lp(P_N')^2d\mu=-2\lambda_N\mathbb{B}\int\lp(P_N')^2+2\lambda_N^2\phi.
\end{align*}
By using (\ref{oef}), we can rewrite this as
\begin{align*}
	(s^2\mathbb{B}+s\mathbb{B}^{(1)}-s\mathbb{A}+2\lambda_N)\mathbb{B}\int\lp(P_N')^2d\mu=2\lambda_N^2\phi.
\end{align*}
Moreover, from \eqref{eq1} and \eqref{eq2}, one knows that
\begin{align*}
(s^2\mathbb{B}+s\mathbb{A}+2\lambda_N)\phi=2\mathbb{B}\int\lp (P_N')^2d\mu,
\end{align*}
thus acting $ (s^2\mathbb{B}+s\mathbb{B}^{(1)}-s\mathbb{A}+2\lambda_N) $ on the above equation, one finally gets
\begin{align*}
	(s^2\mathbb{B}+s\mathbb{B}^{(1)}-s\mathbb{A}+2\lambda_N)(s^2\mathbb{B}+s\mathbb{A}+2\lambda_N)\phi=4\lambda_N^2\phi.
\end{align*}
\begin{proposition}\cite[Corollary 3.2]{ledoux04}\label{P3.1}
The Laplace transform of the (un-normalised) probability density $P_N^2(x) w(x)$ satisfies the differential equation
\begin{align}\label{cle}
(s^4\mathbb{M}_4+s^3\mathbb{M}_3+s^2\mathbb{M}_2+s\mathbb{M}_1)\phi=0,
\end{align}
where $\{\mathbb{M}_i\}_{i=1}^4$ are defined by
\begin{align*}
&\mathbb{M}_1=\mathbb{B}^{(2)}\mathbb{A}-\mathbb{A}^{(2)}\mathbb{B}-\mathbb{A}^{(1)}\mathbb{A}+2\lambda_N\mathbb{B}^{(1)},\\
&\mathbb{M}_2=3\mathbb{B}^{(2)}\mathbb{B}+2\mathbb{B}^{(1)}\mathbb{A}+4\lambda_N\mathbb{B}-\mathbb{A}^2-2\mathbb{A}^{(1)}\mathbb{B},\\
&\mathbb{M}_3=3\mathbb{B}^{(1)}\mathbb{B},\quad \mathbb{M}_4=\mathbb{B}^2.
\end{align*}
\end{proposition}
We give explicit formulas in the appendix.

\subsection{Linear discrete case}
There are many setting involving (\ref{1.1}) defined a linear discrete lattice $\{x\in\mathbb{Z}\}$ in different physically significant models; see for example
\cite{johansson01,johansson03,gorin08,gorin15}.

Ledoux \cite{ledoux05} generalised the considerations revised above for the continuous classical weights 
to the case of discrete classical weights on a linear lattice.
 Starting with a discrete Pearson equation \eqref{dpe}, analogous to (\ref{de})
it can be deduced that the discrete orthogonal polynomials $\{P_N(x)\}_{N=0}^\infty$ satisfy  the second order difference equation \cite{ismail05}
\begin{align*}
[\sigma(x)+\tau(x)]P_N(x+1)-[2\sigma(x)+\tau(x)-\lambda_N]P_N(x)+\sigma(x)P_{N}(x-1)=0.
\end{align*}
Also, with the assumption that the product $\sigma(x) w(x)$ vanishes at the end points of the support, we can check from \eqref{dpe} that the equations
\begin{align*}
\sum_{x\in\mathbb{Z}} \psi(x)\sigma(x) w(x)=\sum_{x\in\mathbb{Z}}\psi(x+1)\sigma(x+1) w(x+1)=\sum_{x\in\mathbb{Z}}\psi(x+1) w(x)(\sigma(x)+\tau(x))
\end{align*}
are valid for any summable function $\psi(x)$. Therefore, if we set $A(x)=\sigma(x)$ and $B(x)=\sigma(x)+\tau(x)$, the above equations can be equivalently written as \cite[Eqs.~(11),(12)]{ledoux05}
\begin{subequations}
\begin{align}
&BP_N(x+1)=(A+B+C_N)P_N(x)-AP_{N}(x-1),\\
&\int Af(x)d\mu=\int Bf(x+1)d\mu,\quad \forall f\in L^1(\mathbb{R})\label{linear}.
\end{align}
\end{subequations}

The equations (\ref{linear}) are fundamental in deriving a difference equation for the Laplace transform of the (un-normalised) measure $P_N^2d\mu$, where $d\mu$ the discrete measure is defined in linear lattice. Since the derivation of the differential equation satisfied by the Laplace transform in discrete case is similar to the continuous case, we omit the details here.

It should be remarked that a general boundary condition was considered in \cite{borodin17}. by setting that the weight function $w (x)$ is not defined on the whole real line, but rather supported by a union of several disjoint intervals with condition 
\begin{align*}
\frac{w(x)}{w(x-1)}=\frac{\phi^+(x)}{\phi^-(x)},
\end{align*}
where $\phi^{\pm}(x)$ vanishes at the ends of each support and are polynomials of degree at most $d$.
This condition was applied to compute the loop equation (or so-called Schwinger-Dyson equation) for several discrete models.

 A closely related loop equation formalism has been applied recently in the study of
 a $q$-boxed plane partition model relating to an example of (\ref{1.1}) on an exponential measure \cite{dimitrov19}.
 This is part of our motivation, in addition to its relevance to moments, 
  to consider whether $q$-orthogonal polynomial ensembles satisfy such an integration by part formula, and whether it can be applied to formulate a $q$-difference equation for the Laplace transform of the density.

\section{Laplace transform of discrete $q$-measure and its $q$-difference equation}\label{S4}
In this part, we consider discrete orthogonal polynomials (or so-called $q$-orthogonal polynomials) defined on the exponential lattice $\{q^n\,|\,n\in\mathbb{Z}\}$. We call a family of $q$-orthogonal polynomials $\{P_N(x;q)\}_{N=0}^\infty$ classical, with respect to the weight function $\rho(x;q)$, if they satisfy the orthogonal relation
\begin{align*}
\int_a^b P_N(x;q)P_M(x;q)\rho(x;q)d_qx=h_N\delta_{N,M},\quad h_N>0,
\end{align*}
where the Jackson's $q$-integral is given as \eqref{16d}.
Moreover, the weight function $\rho(x;q)$ satisfies a $q$-Pearson equation (\ref{16b}) 
with degrees of $\sigma(x)$ and $\tau(x)$ being, at most, of $2$ and $1$ respectively. From \cite{nikiforov86,nodarse06}, one knows that classical $q$-orthogonal polynomials satisfy the following second order difference equation
\begin{align}\label{ee2}
\sigma(x)D_qD_{q^{-1}}P_N(x)+\tau(x)D_qP_N(x)+C_NP_N(x)=0,
\end{align}
where $C_N$ is a constant independent of $x$.
Moreover, regarding with the Pearson pair $(\sigma, \tau)$, one can state the following proposition.
\begin{proposition}
The $q$-Pearson equation (\ref{16b})  induces  the $q$-integration by parts formula
\begin{align}\label{ibp}
\int_a^b \sigma(x) D_qf(x) g(x)d_q\mu=-\int_a^b \tau(x)f(qx)g(qx)d_q\mu-\int_a^b \sigma(x)f(qx)D_qg(x)d_q\mu,
\end{align}
where $d_q\mu=\rho(x;q)d_qx$ and $f(x),\,g(x)$ are two arbitrary $q$-summable functions.
\end{proposition}
\begin{proof}
Let's start with the general integration by parts formula given by \cite{ismail05}
\begin{align*}
\int_a^b \sigma(x)D_q f(x)g(x)d_q\mu= f(x)g(x)\sigma(x)\rho(x)|_a^b -\int_a^b f(qx)D_q(g(x)\sigma(x)\rho(x))d_qx.
\end{align*}
Since for classical $q$-weight functions, $\sigma(x)\rho(x)$ vanishes at the end points of support \cite[eq. (3.6)]{nodarse06}, we know the first term on the right hand side vanishes. 
In relation to the second term on the right hand side, we observe from the $q$-Pearson equation (\ref{16b}) that
\begin{align*}
D_q(g(x)\sigma(x)\rho(x))=D_q g(x)\sigma(x)\rho(x)+g(qx)\tau(x)\rho(x),
\end{align*}
thus completing the proof.
\end{proof}
\begin{remark}
With $q\mapsto q^{-1}$, the above integration by part formula can be written 
\begin{align}\label{ibp2}
\int_a^b \sigma(x)D_{q^{-1}}f(x)g(x)d_q\mu=-q\int_a^b \tau(x)f(x)g(x)d_q\mu-\int_a^b \sigma(x)f(q^{-1}x)D_{q^{-1}}g(x)d_q\mu.
\end{align}
\end{remark}
To obtain a $q$-analogue of the linear difference equation \eqref{linear}, let's introduce a polynomial function
\begin{equation}\label{ee3}
T(x)=\sigma(x)-(1-q)x\tau(x)
\end{equation}
and observe
\begin{align}\label{ibp3}
\int_a^b \sigma(x)f(x)d_q\mu=q\int_a^b T(x)f(qx)d_q\mu.
\end{align}

One can define a $q$-Laplace transform of the un-normalised measure $P_N^2(x)d_q\mu$
\begin{align}\label{ee0}
\phi(\lambda):=\int_a^b e_q(\lambda x)P_N^2(x)d_q\mu,
\end{align}
where $e_q(x)=\sum_{k=0}^\infty ((1-q)x)^k/(q;q)_k$; cf.~(\ref{S.5}). This $q$-exponential function satisfies 
\begin{align}\label{ee1}
D_{q,x}e_q(\lambda x)=\lambda e_q(\lambda x),\quad D_{q^{-1}}e_q(\lambda x)=\lambda e_q(q^{-1}\lambda x)=\lambda \Lambda^{-1}e_q(\lambda x),
\end{align}
where $\Lambda$ is a shift operator such that $\Lambda\phi(\lambda)=\phi(q\lambda)$. 
Therefore, for an arbitrary polynomial function $R(x)=a_nx^n+\cdots+a_0$, 
there exists a corresponding $q$-difference operator $\mathcal{R}(D_{q,\lambda})=a_nD_{q,\lambda}^n+\cdots+a_0$ such that
\begin{align}\label{4.7a}
\mathcal{R}(D_{q,\lambda})\phi(\lambda)=\int_a^b e_q(\lambda x)R(x)P_N^2(x)d_q\mu.
\end{align}
Moreover, for arbitrary $n,m\in\mathbb{N}$, we have
\begin{align*}
\int_a^b e_q(q^n\lambda x)R(q^mx)P_N^2(x)d_q\mu=\Lambda^{n-m}\mathcal{R}(q^mD_{q,\lambda}) \phi(\lambda),
\end{align*}
providing us with an analogue of (\ref{oef}).
 \begin{proposition}
 	For a polynomial $ R(x)=a_nx^n+\cdots+a_0 $,  the corresponding q-difference operator $ \mathcal{R}(D_{q,\lambda}) $ satisfies
 	\begin{align*}
 		\mathcal{R}(D_{q,\lambda})(\lambda^m\phi)=\sum^{m}_{i=0}\lambda^{m-i}\left[\begin{array}{c}
 		m\\
 		i
 		\end{array}\right]_q\mathcal{R}^{(i)}_{m}\phi,\quad m\in\mathbb{Z}_{\geq0},
 	\end{align*}
	where $\mathcal{R}_{m}^{(i)}$ is the operator corresponding to $ \Lambda^{m-i}D_{q}^iR(x)$.
	\end{proposition}
 	\begin{proof}
 		Let's first prove that
		\begin{align}\label{leibniz}
		D_q^n (\lambda f)=[n]_qD_q^{n-1}f+q^n \lambda D_q^n f,\quad [n]_q=\frac{1-q^n}{1-q}
		\end{align}
		holds for arbitrary function $f$. Making use of induction, we assume that \eqref{leibniz} is true and then 
		\begin{align*}
		D_q^{n+1}(\lambda f)&=D_q([n]_qD_q^{n-1}f+q^{n}\lambda D_q^{n}f)\\&=[n]_qD_q^{n}f+q^{n}D_q^{n}f+q^{n+1}\lambda D_q^{n+1} f\\&=[n+1]_qD_q^n f+q^{n+1}\lambda D_q^{n+1}f.
		\end{align*}
Taking $f=\lambda^{m-1}\phi$, $m\geq 1$, and using $q$-binomial formula the result follows.
 	\end{proof}
	For example, if we consider $R(x)$ as a second order polynomial, we have
\begin{align*}
\mathcal{R}(\lambda \phi)&=\lambda \mathcal{R}^{(0)}_1\phi+\mathcal{R}_1^{(1)}\phi,\\
\mathcal{R}(\lambda^2\phi)&=\lambda^2\mathcal{R}_2^{(0)}\phi+\lambda (1+q)\mathcal{R}_2^{(1)}\phi+\mathcal{R}_2^{(2)}\phi,\\
\mathcal{R}(\lambda^3\phi)&=\lambda^3\mathcal{R}_3^{(0)}\phi+\lambda^2 (1+q+q^2)\mathcal{R}_3^{(1)}\phi+\lambda (1+q+q^2)\mathcal{R}_3^{(2)}\phi
\end{align*}
since $\mathcal{R}_3^{(3)}=0$. Specifically, if we denote $R(x)=r_2x^2+r_1x+r_0$, then we have
\begin{align*}
&\mathcal{R}_1^{(0)}=r_2q^2D_q^2+r_1qD_q+r_0,\quad \mathcal{R}_1^{(1)}=r_2(1+q)D_q+r_1,\\
&\mathcal{R}_2^{(0)}=r_2q^4D_q^2+r_1q^2D_q+r_0,\quad \mathcal{R}_2^{(1)}=r_2q(1+q)D_q+r_1,\quad \mathcal{R}_2^{(2)}=r_2(1+q),\\
&\mathcal{R}_3^{(0)}=r_2q^6D_q^3+r_1q^3D_q+r_0,\quad \mathcal{R}_3^{(1)}=r_2q^2(1+q)D_q+r_1,\quad \mathcal{R}_3^{(2)}=r_2(1+q).
\end{align*}

For later use, we denote $\mathcal{A}$, $\mathcal{B}$ and $\mathcal{T}$ as operators corresponding to $\tau(x)$, $\sigma(x)$ and $T(x)$
in the sense of (\ref{4.7a}).
Now taking $f=e_q(\lambda x)$ and $g=P_N^2(x)$, and using formulas \eqref{ibp2} and \eqref{ibp3}, one can get
\begin{align}\label{qeq1}
(\lambda \Lambda^{-1}\mb +q  \ma)\phi=-\Lambda^{-1}\mb \int_0^\infty e_q(\lambda x)P_ND_{q^{-1}}P_Nd_q\mu-q\mt \int_0^\infty e_q(\lambda x)P_N D_q P_Nd_q\mu.
\end{align}
Since $q\mt\int_0^\infty e_q(\lambda x)P_ND_qP_Nd_q\mu=\mb \int_0^\infty e_q(q^{-1}\lambda x)P_N(q^{-1}x)D_{q^{-1}}P_Nd_q\mu$, and by using formula \eqref{ibp} with $f=e_q(\lambda x)$ and $g=P_N(q^{-1}x)D_{q^{-1}}P_N$, one has
\begin{align}\label{qeq2}
\lambda\mt\int_0^\infty e_q(\lambda x)P_ND_qP_Nd_q\mu=C_N\phi-q^{-1}\mb \int_0^\infty e_q(\lambda x)(D_{q^{-1}}P_N)^2d_q\mu.
\end{align}
Similarly, by taking $ f=e_q(\lambda x) $ and $ g=P_ND_{q^{-1}}P_N $, one has
\begin{align}\label{qeq3}
	\lambda\mb\int_{0}^{\infty}e_q(\lambda x)P_ND_{q^{-1}}P_Nd_q\mu=C_N\Lambda\phi-q^{-1}\mb\int_0^\infty e_q(\lambda x)(D_{q^{-1}}P_N)^2d_q\mu.
\end{align}
Use  of formula \eqref{ibp} is made to evaluate the last term in the above equation. Thus by taking $f=e_q(\lambda x)$ and $g=(D_{q^{-1}}P_N)^2$ in the integration by parts formula, one has
\begin{align}\label{qeq4}
\lambda \mb \int_0^\infty &e_q(\lambda x)(D_{q^{-1}}P_N)^2d_q\mu=\Lambda \ma \int_0^\infty e_q(\lambda x)D_qP_ND_{q^{-1}}P_Nd_q\mu\nonumber\\
&+C_N\Lambda \int_0^\infty e_q(\lambda x)P_ND_{q^{-1}}P_Nd_q\mu
+C_N\Lambda \int_0^\infty e_q(\lambda x)P_ND_qP_Nd_q\mu.
\end{align}
Here we remark that when we take the continuum limit $q\to 1$, then equations \eqref{qeq1}-\eqref{qeq4} coincide with equations \eqref{eq1}-\eqref{eq3}. However, these equations have not yet been closed --- more relations are needed.
In this regard, one can notice that
\begin{align*}
\sigma D_{q^{-1}}P_N(D_qP_N-D_{q^{-1}}P_N)=\sigma (q-1)xD_{q^{-1}}P_ND_qD_{q^{-1}}P_N=(1-q)xD_{q_{-1}}P_N(\tau D_qP_N+C_NP_N),
\end{align*}
which results in 
\begin{align}\label{qeq5}
\mt \int_0^\infty e_q(\lambda x)D_{q^{-1}}P_ND_qP_Nd_q\mu=\mb \int_0^\infty e_q(\lambda x)(D_{q^{-1}}P_N)^2d_q\mu+C_N\mf \int_0^\infty e_q(\lambda x)P_ND_{q^{-1}}P_Nd_q\mu,
\end{align}
where $\mf$ is the $q$-difference operator corresponding to $(1-q)x$ satisfying $\mt=\mb-\ma\mf$. 
Therefore, we end up with five equations (\eqref{qeq1}-\eqref{qeq5}) for $ \phi $, $\psi_1,\,\psi_2,\,\psi_3$ and $\psi_4$, where
\begin{align*}
 &\psi_1:=\int_0^\infty e_q(\lambda x)P_ND_{q^{-1}}P_Nd_q\mu,\quad \psi_2:=\int_0^\infty e_q(\lambda x)P_ND_qP_Nd_q\mu,\\
 & \psi_3:=\int_0^\infty e_q(\lambda x)(D_{q^{-1}}P_N)^2d_q\mu,\quad \psi_4:=\int_0^\infty e_q(\lambda x)D_{q^{-1}}P_ND_qP_Nd_q\mu.
\end{align*}

 We will now show how one can eliminate the latter four terms to get an equation for $ \phi $. 
By acting operator $ \mt\Lambda^{-1} $ on both sides of \eqref{qeq4} and making use of \eqref{qeq5}, one has
\begin{align*}
(q^{-1}\mt\lambda\Lambda^{-1}-\ma)\mb\psi_3=C_N\mb\psi_1+C_N\mt\psi_2
\end{align*}
By substituting equation \eqref{qeq2} and \eqref{qeq3} into the above equation, one has
\begin{align}\label{qeq6}
	(\lambda q^{-1}\mt\lambda\Lambda^{-1}-\lambda\ma+2q^{-1}C_N)\mb\psi_3=C_N^2(1+\Lambda)\phi.
\end{align}
Meanwhile, by substituting equation \eqref{qeq2} and \eqref{qeq3} into equation \eqref{qeq1}, one finds that
\begin{align*}
	q(q\lambda^2\mb+q\lambda\Lambda\ma+2C_N\Lambda)\phi=(1+\Lambda)\mb\psi_3.
\end{align*}
Thus, together with equation \eqref{qeq6}, we have
\begin{align}\label{qle}
	(\lambda\mt\lambda\Lambda^{-1}-q\lambda\ma+2C_N)(1+\Lambda)^{-1}(q\lambda^2\mb+q\lambda\Lambda\ma+2C_N\Lambda)\phi=C_N^2(1+\Lambda)\phi.
\end{align}
By noting that for $i\in\mathbb{Z}$, there holds a commutative formula
\begin{align*}
(1+\Lambda)^{-1}\lambda^i=\lambda^i(1+q^i\Lambda)^{-1},
\end{align*}
and degree of $\mathcal{T}$ and $\ma$ are no greater than 2 and 1 respectively,
we can state the following result.

\begin{theorem}\label{prop4}
	The $q$-Laplace transform of the un-normalised measure $P_N^2d_q\mu$ on the exponential lattice satisfies 
	\begin{align}\label{qlaplace}
	(\lambda^4\mathcal{M}_4+\lambda^3\mathcal{M}_3+\lambda^2\mathcal{M}_2+\lambda\mathcal{M}_1+\mathcal{M}_0)\phi=0,
	\end{align}	where 
	\begin{align*}
	&\mathcal{M}_4=q^{-1}\mt_3^{(0)}\Lambda^{-1}(1+q^2\Lambda)^{-1}\mb,\\
	&\mathcal{M}_3=(q^{-1}(1+q+q^2)\mt_3^{(1)}\Lambda^{-1}-q^2\ma_2^{(0)})(1+q^2\Lambda)^{-1}\mb+\mt_2^{(0)}(1+q\Lambda)^{-1}\ma,\\
	&\mathcal{M}_2=(q^{-1}(1+q+q^2)\mt_3^{(2)}\Lambda^{-1}-q^2(1+q)\ma_2^{(1)}+2qC_N)(1+q^2\Lambda)^{-1}\mb\\&\qquad\qquad+((1+q)\mt_2^{(1)}-q^2\ma_1^{(0)}\Lambda)(1+q\Lambda)^{-1}\ma+2C_N\mt_1^{(0)}(1+\Lambda)^{-1},\\
	&\mathcal{M}_1=(\mt_2^{(2)}-q^2\ma_1^{(1)}\Lambda+2qC_N\Lambda)(1+q\Lambda)^{-1}\ma+2C_N(\mt_1^{(1)}-q\ma\Lambda)(1+\Lambda)^{-1},\\
	&\mathcal{M}_0=C_N^2(4(1+\Lambda)^{-1}\Lambda-(1+\Lambda)).
	\end{align*}
	Here $\Lambda$ is the shift operator satisfying $\Lambda \phi(x)=\phi(\lambda x)$, $C_N$ is a particular coefficient in the $q$-difference equation (\ref{ee2}) and 
	$\mathcal{A}$, $\mathcal{B}$ and $\mathcal{T}$ are $q$-difference operators related to $\tau(x)$, $\sigma(x)$ and $T(x)$ in the sense of (\ref{4.7a}), 
	with $\sigma(x), \tau(x)$ specified as the Pearson pair and $T(x)=\sigma(x)-(1-q)\tau(x)$. 
	\end{theorem}

\section{Application in the case of the discrete $q$-Hermite weight}\label{S5}

Let $\{ p^{(d\text{-}qH)}_n(x;q) \}_{n\in \mathbb{N}}$ denote the monic $q$-orthogonal polynomials with respect to the
discrete $q$-Hermite weight (\ref{24}). 
It has been shown in \cite{Wi12} that up to proportionality, the moments of the corresponding un-normalised
measure $( p^{(d\text{-}qH)}_N(x;q))^2 d_q \mu$ have combinatorial meaning in terms of rook placements. This motivates giving
particular attention to this case.


Different from \cite[Sec. 3.28]{koekoek96}, we consider orthonormal polynomials $\{H_\ell(x;q)\}_{\ell\in\mathbb{N}}$ satisfying
\begin{align*}
\int_{-1}^1 H_m(x;q)H_n(x;q)d_q\mu=\delta_{n,m},
\end{align*}
where $ d_q\mu $ is specified by the corresponding weight function of discrete $q$-Hermite  polynomials given by \eqref{24}.
We first make explicit the $q$-difference equation \eqref{qlaplace} for the $q$-Laplace transform
of the un-normalized measure
$(H_N(x;q))^2d_q\mu$. For simplicity, we consider the $N=0$ case, i.e. 
\begin{equation}\label{5.1s}
	\phi_0(\lambda) := \int_{-1}^1 e_q(\lambda x) w^{(d\mbox{-}qH)}(x) \, d_qx.
\end{equation}
Noting in the definition of $e_q(x)$ as a sum below (\ref{ee0}) that all terms are positive, as is $w^{(d\mbox{-}qH)}(x)$,
 the order of the sum and the integral can be interchanged.  Computing the resulting integral using (\ref{21d}) shows
\begin{equation}\label{5.1t}
	\phi_0(\lambda) = \sum_{k=0}^{\infty}\frac{(1-q)^{2k+1}\lambda^{2k}}{(q^2;q^2)_k}=(1-q)e_{q^2}\left(
	\frac{\lambda^2}{(1+q)^2}
	\right).
\end{equation}
The simplicity of (\ref{5.1t}) also allows for an independent verification of the $q$-difference equation.

\begin{proposition}
	The $q$-difference Laplace transform $\phi_0(\lambda)$ satisfies the q-difference equation
	\begin{align}\label{5.3a}
		&\left(\lambda^4 q^{-1}\Lambda^{-1}(1+q^2\Lambda)^{-1}(D_q^2-1)\right. \nonumber \\
		&+\lambda^3(q^4(1-q)^{-1}D_q(1+q^2\Lambda)^{-1}(D_q^2-1)-(1-q)^{-1}(1+q\Lambda)^{-1}D_q)   \nonumber \\
		&+\lambda^2(q^2(1+q)(1-q)^{-1}(1+q^2\Lambda)^{-1}(D_q^2-1)+q^3(1-q)^{-2}D_q\Lambda(1+q\Lambda)^{-1}D_q )   \nonumber \\
		&+\lambda q^2(1-q)^{-2}\Lambda(1+q\Lambda)^{-1}D_q)\phi=0,
	\end{align}	
\end{proposition}
\begin{proof}
	The Pearson pair with respect to the discrete $ q $-Hermite polynomials is given by
\begin{align}\label{pp-dqH}
(\sigma,\tau)=(x^2-1,(1-q)^{-1}x),
\end{align}
and the corresponding operators in the sense of (\ref{4.7a}) are
\begin{align*}
A(x)=(1-q)^{-1}x,\quad B(x)=x^2-1,\quad T(x)=-1,
\end{align*}
while the constant is given by $C_0=0$.
By substituting these into Proposition \ref{prop4}, we obtain 
the explicit formula for those difference operators and thus the q-difference equation. 

On the other hand, the validity of (\ref{5.3a}) can be checked directly.
By denoting
\begin{align*}
f_k=\frac{(1-q)^{2k+1}}{(q^2;q^2)_k}, \quad [k]_q=\frac{1-q^k}{1-q},
\end{align*}
and taking $\phi_0(\lambda)$ in terms of the series form from (\ref{5.1t})
 and substituting it into the q-difference equation, one can compute the coefficients of $ \lambda^{2k} $ as
\begin{align*}
	[\lambda^{2k}]&=\frac{1}{1+q^{2k-2}}\left(-q^{-2k+3}[2k-2]_q[2k-3]_qf_{k-1}+q^{-2k+3}f_{k-2}-\frac{[2k-2]_q}{1-q}f_{k-1}\right)\\
	&+\frac{1}{1+q^{2k}}\left(-\frac{q^4[2k]_q[2k-1]_q[2k-2]_q}{1-q}f_k+\frac{q^4[2k-2]_q}{1-q}f_{k-1}\right.\\
	&\left.-\frac{q^2(1+q)}{1-q}\left([2k]_q[2k-1]_qf_k-f_{k-1}\right)-\frac{q^{2k+2}[2k]_q[2k-1]_q}{(1-q)^2}f_k-\frac{q^{2k+1}[2k]_q}{(1-q)^2}f_k\right)\\
	&=\frac{q^{2k-1}(1-q^{2k})f_k}{(1-q)^4(1+q^{2k})}[q^4(1-q^{2k-2})+q^2(1-q^2)-q^3(1-q^{2k-1})-q^2(1-q)]=0,
\end{align*}
as required.
\end{proof}
We now turn to the explicit formula for the moments of $q$-Hermite polynomials by making use of the Pearson relation.
It is known that the normalized discrete q-Hermite polynomials have the  generating function \cite{koekoek10}
\begin{align*}
	\frac{\left(t^{2} ; q^{2}\right)_{\infty}}{(x t ; q)_{\infty}}=\sum_{n=0}^{\infty} \frac{q^{\frac{1}{4}n(n-1)}\sqrt{1-q}H_{n}(x ; q)}{\sqrt{(q ; q)_{n}}} t^{n}.
\end{align*}
The generating function implies the $q$-difference equation 
\begin{align}\label{5.4a}
	H_n(qx)=H_n(x)-(1-q^n)d_n^{-1}xH_{n-1}(x),\quad d_n=q^{(n-1)/2}(1-q^n)^{1/2},
\end{align}
and the three term recurrence relation
\begin{align*}
xH_n(x)=d_{n+1}H_{n+1}(x)+d_nH_{n-1}(x).
\end{align*}
Combining these two equations leads to 
\begin{align}\label{5.4b}
H_n(qx)=q^n H_n(x)-q^{-n+1}d_nd_{n-1}H_{n-2}(x),
\end{align}
which is useful for deriving a particular linearisation formula.
\begin{proposition}
The shifted  normalised discrete $q$-Hermite  polynomial $H_n(q^kx)$ is expressed in terms
of the un-shifted normalised  polynomials according to
\begin{align*}
	H_n(q^kx)=\sum_{l=0}^{\min(k,[\frac{n}{2}])}\left((-1)^{l}q^{(k-l)(n-2l)+l(l-n)}\begin{aligned}
		\left[\begin{array}{c}
			k\\
			k-l
		\end{array}\right]
	\end{aligned}_{q^2}\prod_{i=0}^{2l-1}d_{n-i}\right)H_{n-2l}(x).
\end{align*}
\end{proposition}

\begin{proof}
By induction using (\ref{5.4b}), a combinatorial formula for $H_n(q^k x)$ is
\begin{align*}
H_n(q^k x)=\sum_{l=0}^{\min(k,\left[
\frac{n}{2}
\right])}\left((-1)^{l}q^{l(l-n)}\sum_{p\in\mathcal{P}_l^{k-l}}q^{\sum_{i=1}^l p_i(n-2i+2)}\prod_{i=0}^{2l-1}d_{n-i}\right)H_{n-2l}(x),
\end{align*}
where $ \mathcal{P}_{l}^m $ denotes all the $ l $-partitions $ (p_1,p_2,\dots,p_l) $ of $m$ such that $ \sum_{i=1}^l p_i=m $. This summation can be interpreted as a generating function for paths between two fixed points in the $ \mathbb{Z}^2 $ lattice, which is explained below. Denote 
\begin{align*}
 f_{k,n}^l=\sum_{p\in\mathcal{P}_{l}^{k}}q^{\sum_{i=1}^l p_i(n-2i+2)}:=q^{k(n-2l+2)}\hat{f}_k^l
\end{align*}
and consider the lattice $ \mathbb{Z}\times 2\mathbb{Z} $ with edges pointing rightward and downward, one can view $ \hat{f}_k^l $ as the generating function of paths from the point $ (0,2(l-1)) $ to the point $ (k,0) $, and one obtains the  recurrence relation
\begin{align*}
	q^{2k}\hat{f}_k^{l-1}+\hat{f}^l_{k-1}=\hat{f}_k^l,\qquad
	\hat{f}_1^1=1,\quad \hat{f}_1^2=1+q^2,\quad \hat{f}_2^1=1.
\end{align*}
The recurrence relation is solved in terms of the $q$-binomial formula and
\begin{align*}
	\hat{f}^l_k=\begin{aligned}
		\left[\begin{array}{c}
			l-1+k\\
			k
		\end{array}\right]
	\end{aligned}_{q^2},
\end{align*}
 thus completing the proof.
\end{proof}

We are now in a position to deduce an evaluation of the moments of the normalised measure $ H_N^2(x)
\, d_q\mu$.

\begin{proposition}
Define
\begin{align}\label{26.0}
\tilde{M}_{k,N} :=\int_{-1}^1 \prod_{i=0}^{k-1}(1-q^{-2i}x^2) 
 H_N^2(x) 
\, d_q\mu,
\qquad 
\tilde{m}_{2p,N}:=\int_{-1}^1 x^{2p} 
 H_N^2(x)
\, d_q\mu.
\end{align}
One has
\begin{equation}\label{26.1}
\tilde{M}_{k,N}
= q^k \sum_{l=0}^{k} q^{2kN+l(4l-4k-2N-1)}\begin{aligned}
		\left[\begin{array}{c}
			k\\
			k-l
		\end{array}\right]
	\end{aligned}_{q^2}^2\frac{(q;q)_N}{(q;q)_{N-2 l}}
	\end{equation}
and
\begin{equation}\label{26.2}	
\tilde{m}_{2p,N}=\sum_{i=0}^p (-1)^i q^{i(i-1)}\left[
\begin{array}{c}
p\\
i
\end{array}
\right]_{q^{2}} \tilde{M}_{i,N}.
\end{equation}
\end{proposition}

\begin{proof}
From the integration by parts formula \eqref{ibp3}, one notices that for the measure $ d_q\mu $ of the discrete $q$-Hermite  polynomials, 
\begin{align*}
q\int_{-1}^1 f(qx)d_q\mu=\int_{-1}^1 (1-x^2)f(x)d_q\mu.
\end{align*}
Hence
\begin{multline}
	\int_{-1}^1\prod_{i=0}^{k-1}(1-q^{-2i}x^2)H_N^2(x)d_q\mu=q^k\int_{-1}^1H_N^2(q^kx)d_q\mu\\
	=q^k\sum_{l=0}^{\min(k,[\frac{N}{2}])}q^{2(k-l)(N-2l)+2l(l-N)}\begin{aligned}
		\left[\begin{array}{c}
			k\\
			k-l
		\end{array}\right]
	\end{aligned}_{q^2}^2\prod_{i=0}^{2l-1}d^2_{N-i},
\end{multline}
which when substituted in the definition of $\tilde{M}_{k,N}$ from (\ref{26.0}), 
and upon recalling the definition of $d_n$ from (\ref{5.4a}), gives (\ref{26.1}).

By recognising the fact that
\begin{align*}
x^{2p}=\sum_{i=0}^p (-1)^i q^{i(i-1)}\left[
\begin{array}{c}
p\\
i
\end{array}
\right]_{q^{2}}\prod_{j=0}^{i-1}(1-q^{-2j}x^2),
\end{align*}
we see from the definition of $\tilde{m}_{2p,N}$ in (\ref{26.0}) that
the evaluation (\ref{26.2}) follows as a corollary of (\ref{26.1}).
\end{proof}

\begin{remark}
From the definition of the spectral density moments $m_{2k,N}^{(d\mbox{-}qH)}$ implied by substituting (\ref{1.2c}) in
(\ref{1.1c}), and the definition of $\tilde{m}_{2p}$ from (\ref{26.0}) we have 
\begin{equation}\label{5.9a}
\tilde{m}_{2p,N} = m_{2p,N+1}^{(d\mbox{-}qH)} - m_{2p,N}^{(d\mbox{-}qH)}.
\end{equation}
Substituting for the RHS according to (\ref{30e}) gives a formula for $\tilde{m}_{2p,N}$ involving the difference of
two single sums rather than a double sum as in (\ref{26.2}). For low orders we have used computer algebra to check that both expressions
are  identical, but do not have a general proof.

The expression for $\tilde{m}_{2p,N}$ as implied by (\ref{5.9a}) is already reported in \cite[Th.~5]{PV14}. Previously, using
combinatorial methods based on a rook placement interpretation of $\{ \tilde{m}_{2p,N}  \}$ yet another distinct expression
was obtained \cite[Eq.~(5.26)]{Wi12}.

Another relevant point is that $q$-difference equation  \eqref{qlaplace} for the $q$-Laplace transform implies a
corresponding linear recurrence for the moments, analogous to what is known in the continuous classical cases
\cite{ledoux04}.
\end{remark}

\subsection*{Acknowledgements}
	The work of PJF is part of the program of study supported
	by the Australian Research Council Centre of Excellence ACEMS
	and the Discovery Project grant DP210102887. SHL and GYF are supported by National Natural Science Foundation of China with grants NSFC12175155, NSFC11871336.

\appendix
\section*{Appendix}\label{appendix}

Here we  rederive known differential equations satisfied by the Laplace transform of the un-normalized measure $(P_N(x))^2 d \mu$
in the classical continuous cases using the Pearson pair formalism
of Section \ref{S3.2}. The classical continuous cases correspond to the weights (\ref{1.1b}), distinguished by having all moments
finite, and the Cauchy weight
$
\left(1+x^{2}\right)^{-a}$, $a > 1/2$,
where the number of finite moments depends on $a$.
The corresponding Pearson pairs are
\begin{align*}
(\sigma,\tau)=\left\{\begin{array}{ll}
(1,-2x)&\text{Hermite case,}\\
(x,a+1-x)&\text{Laguerre case,}\\
(x(1-x),a+1-(a+b+2)x)&\text{Jacobi case,}\\
(1+x^2,2(1-a)x)&\text{Cauchy case}
\end{array}
\right.
\end{align*}
According to the theory of Section \ref{S3.2}, these imply the differential
operators
\begin{align*}
&\mathbb{A}_H=-2\frac{d}{ds},\quad\mathbb{A}^{(1)}_H=-2,\quad\mathbb{B}_H=1,\quad\mathbb{B}^{(1)}_H=\mathbb{B}^{(2)}_H=0,\quad \lambda_N^{(H)}=2N,\\
&\mathbb{A}_L=a+1-\frac{d}{ds},\quad\mathbb{A}^{(1)}_L=-1,\quad\mathbb{B}_L=\frac{d}{ds},\quad\mathbb{B}_L^{(1)}=1,\quad\mathbb{B}_L^{(2)}=0,\quad \lambda_N^{(L)}=N,\\
&\mathbb{A}_J=a+1-(a+b+2)\frac{d}{ds},\quad\mathbb{A}_J^{(1)}=-a-b-2,\quad\mathbb{B}_J=\frac{d}{ds}-\frac{d^2}{ds^2},\quad\mathbb{B}_J^{(1)}=1-2\frac{d}{ds},\\
&\qquad\mathbb{B}_J^{(2)}=-2,\quad \lambda_N^{(J)}=N(N+a+b+1),\\
&\mathbb{A}_C=2(1-a)\frac{d}{ds},\quad\mathbb{A}_C^{(1)}=2(1-a),\quad\mathbb{B}_C=1+\frac{d^2}{ds^2},\quad\mathbb{B}_C^{(1)}=2\frac{d}{ds},\quad\mathbb{B}_C^{(2)}=2,\quad\\
&\qquad \lambda_N^{(C)}=-N(N-2a+1)
\end{align*}
in the four cases respectively.
Application of Proposiiton \ref{P3.1} then gives the
differential equations of the Laplace transforms of the un-normalized measure.
These are  
\begin{align*}	
4s\Phi^{''}+4\Phi^{'}-(s^3+(8N+4)s)\Phi=0,
\end{align*}
in Hermite case
\begin{align*}	
	(s^3-s)\Phi^{''}+(3s^2+2(2N+a+1)s-1)\Phi^{'}+(s(1-a^2)+a+1+2N)\Phi=0,
\end{align*}
in the  Laguerre case
\begin{align*}
&s^3\Phi^{(4)}+(6s^2-2s^3)\Phi^{'''}+\left(s^3-9s^2+[6-4N(N+a+b+1)-(a+b)(a+b+2)]s\right)\Phi^{''}\\
	&+\left(3s^2+[4N(N+a+b+1)+2(a+1)(a+b)-6]s-(a+b)(a+b+4)-4N(N+a+b+1)\right)\Phi^{'}\\
	&+\left((1-a^2)s+(a+b)(a+1)+2N(N+a+b+1)\right)\Phi=0
\end{align*}
in the Jacobi case, and
\begin{align*}
	\lambda^3\Phi^{(4)}&+6\lambda^2\Phi^{'''}+\left(2\lambda^3+[4a(1-a)+6-4N(N-2a+1)]\lambda\right)\Phi^{''}\\
	&+\left(6\lambda^2+4[a(1-a)-N(N-2a+1)]\right)\Phi^{'}+\left(\lambda^3+[2-4N(N-2a+1)+4a]\lambda\right)\Phi=0,
\end{align*}
in the Cauchy case. 
The first three cases was considered in \cite{ledoux04} and the Cauchy one was recently shown in \cite{assiotis21}; see also \cite{FR21} in relation to the latter.

\end{document}